\documentclass[lettersize,journal]{IEEEtran}
\usepackage{amsmath,amsfonts}
\usepackage{array}
\usepackage{textcomp}
\usepackage{stfloats}
\usepackage{url}
\usepackage{verbatim}
\usepackage{graphicx}
\usepackage{cite}
\usepackage[caption=false,font=normalsize,labelfont=sf,textfont=sf]{subfig}
\usepackage{amssymb,amsthm}
\usepackage{xcolor}
\usepackage{booktabs}
\usepackage{multirow}
\usepackage{multicol}
\usepackage[ruled,linesnumbered]{algorithm2e}
\usepackage[para,flushleft]{threeparttable}

\usepackage[hidelinks]{hyperref}

\newtheorem{definition}{Definition}
\newtheorem{theorem}{Theorem}
\newtheorem{lemma}{Lemma}
\SetFuncSty{text}

\definecolor{myred}{RGB}{242, 220, 218}
\definecolor{myblue}{RGB}{219, 238, 243}

\definecolor{lv}{RGB}{129, 178, 31}
\definecolor{danlan}{RGB}{171, 215, 236}
\definecolor{lan}{RGB}{51, 171, 193}
\definecolor{hong}{RGB}{245, 155, 123}
\definecolor{cheng}{RGB}{237, 136, 40}
\definecolor{danhong}{RGB}{249, 179, 173}
\begin{document}

\title{FastReChain: A Novel Bidirectional Model-Based Algorithm for Topology Engineering of OCS-Based Clusters}

\author{
  \IEEEauthorblockN{
    Zihan Zhu\IEEEauthorrefmark{8}$^*$,
    Xinchi Han\IEEEauthorrefmark{6}$^*$,
    Dongchao Wu\IEEEauthorrefmark{5},
    Zhanbang Zhang\IEEEauthorrefmark{5},
    Jian Yang\IEEEauthorrefmark{5},
    Shizhen Zhao\IEEEauthorrefmark{6},
    Xinbing Wang\IEEEauthorrefmark{6}
  }

  \IEEEauthorblockA{
    $^*$These authors should be considered co-first authors.\\
    \IEEEauthorrefmark{8} University of Electronic Science and Technology of China, Chengdu, China\\
    \IEEEauthorrefmark{6} Shanghai Jiao Tong University, Shanghai, China\\
    \IEEEauthorrefmark{5} Huawei, Dongguan, China
  }

  \IEEEauthorblockA{
    2021080911004@std.uestc.edu.cn,
    hanxinchi@sjtu.edu.cn,
    wudongchao@huawei.com,
    zhangzhanbang1@huawei.com,
    yangjian227@huawei.com,
    shizhenzhao@sjtu.edu.cn,
    xwang8@sjtu.edu.cn
  }
}

\markboth{Transactons on Networking,~Vol.~XX, No.~XX, February~2026}%
{Shell \MakeLowercase{\textit{et al.}}: A Sample Article Using IEEEtran.cls for IEEE Journals}


\maketitle

%-------------------------------------------------------------------------------
\begin{abstract}
Optical Circuit Switching (OCS) technology is increasingly being adopted in data centers due to its advantages of low power consumption and low technology refresh costs. Unlike electrical packet switches, OCS provides programmable bandwidth for directly connected devices by configuring the mapping relationships of internal ports. Thus, how to calculate these internal port mapping relationships, i.e., Topology Engineering (ToE), is one of the key designs of OCS-based clusters.

Current deployments usually design ToE algorithms by solving Integer Linear Programming (ILP) models, with the aim of minimizing modifications to links occupied by running tasks as much as possible. However, ILP-based ToE algorithms may incur excessive runtime overhead in large-scale clusters. Some existing ToE algorithms convert the ILP model into a Minimum-Cost Flow model through greedy construction, yet such greedy strategies may increase the number of affected links during the OCS reconfiguration process. To solve the aforementioned problems, we propose a novel bidirectional modeling approach, along with a corresponding \emph{FastReChain} algorithm in this paper. We verify the superiority of this algorithm through simulation experiments based on real-trace data.

% In this paper, we propose a centralized scheduling algorithm that achieves theoretical maximum throughput even in one-rate bidirectional Clos networks, while producing schemes with near-minimal numbers of rearrangements. It is the only algorithm that directly supports bidirectional Clos networks and has a time efficiency high enough to support dynamic scheduling to date. For static minimal rewiring, its running time ranges from a fraction to a few hundredths of other algorithms, and the number of rearrangements has also been steadily improved, allowing for more frequent adjustments and less impact on ongoing communications. In addition, the algorithm is very flexible and can support various functional requirements in real-world environments. We achieve this result through the replacement chain concept and bitset optimization.
\end{abstract}


%-------------------------------------------------------------------------------

\section{Introduction}

\IEEEPARstart{I}{n} recent years, the growing computational demands of applications such as Data Center Networks (DCNs) and large language model (LLM) training have driven the adoption of Optical Circuit Switches (OCSes) to replace core-layer electrical switches. This shift offers benefits including lower technology refresh costs, improved scalability, and reduced operational power consumption, as demonstrated in systems such as Jupiter Evolving, TPUv4, and Google A3~\cite{poutievski2022jupiter,zu2024resiliency,GoogleA3Supercomputers,2021Gemini}. Leveraging the reconfigurable link flexibility of OCSes, OCS-based clusters typically compute the required link provisioning between directly connected electrical switches, which is \emph{a.k.a.} \emph{logical topology}, according to traffic patterns. Once the logical topology is determined, we need to further determine the port matching relationships within OCSes to produce the circuit-level links, thereby providing the bandwidth required by the logical topology for directly connected devices, which is \emph{a.k.a.} \emph{topology engineering (ToE)}. Therefore, the design of the ToE algorithm is one of the key aspects of OCS-based cluster. Two key metrics can be used to evaluate the quality of a ToE algorithm:

% \hxc{
% The frequency of ToE execution varies across different deployment scenarios. For example, in DCN workloads like those handled by Jupiter Evolving~\cite{}, ToE may be performed on a daily or weekly basis. In OCS-based GPU clusters such as TPUv4~\cite{} and TopoOPT~\cite{}, it may be triggered upon the arrival of each new task. In clusters with more dynamic traffic patterns, such as MixNet~\cite{}, ToE may need to be executed as often as every few hundred milliseconds. Two features can be considered to assess the quality of a ToE algorithm: the rewiring ratio and the solving time of ToE.}

\textsf{\bfseries Rewiring Ratio.} \emph{Rewiring Ratio} quantifies the number of circuit paths occupied by current running tasks that must be reconfigured when performing ToE. Since OCS reconfiguration may disrupt communication of currently running tasks and reduce its efficiency \cite{han2024lumoscore}, it is desirable to minimize Rewiring Ratio for better communication throughput. The \emph{Rewiring Ratio} is quantified and analyzed in detail in Section~\ref{evalsetup}.

\textsf{\bfseries Solving Overhead.} \emph{Solving Overhead} is defined as the execution latency of the ToE algorithm. For OCS-based Machine Learning (ML) clusters such as TopoOpt and TPUv4, task-level logical topology reconfiguration may be required to meet the communication demands of ML training jobs. Excessively high solving overhead of the ToE algorithm can degrade the resource utilization efficiency of the cluster.

Previous works fail to achieve optimality in terms of the aforementioned features. Commercially deployed OCS-based clusters typically design ToE algorithms based on Integer Linear Programming (ILP) models \cite{poutievski2022jupiter,zu2024resiliency}; while ILP achieves optimal rewiring ratios, it incurs excessively high solving overhead in large-scale clusters. To mitigate this overhead, Google and Zhao et al. designed ToE algorithms based on the Minimum-Cost Flow (MCF) model \cite{zhao2019minimal,han2024lumoscore,zhang2023reducing}. However, this approach has two drawbacks: first, the MCF model itself is difficult to directly express the layer 2 (L2) symmetry constraint\footnote{A constraint that requires OCSes to provide bidirectional links for directly connected devices; otherwise, the L2 protocol cannot work normally \cite{han2024lumoscore,EthernetAddressResolution1982}}. Previous studies realized the satisfaction of L2 constraints by greedily reducing the problem, along with the logical topology to non-L2 constrained ones \cite{han2024lumoscore}, but this reduction itself may affect the minimization of the rewiring ratio as we will discuss in Section~\ref{basicmodel}; second, in such methods, the number of connections between every OCS and directly connected devices needs to be a fixed value, which may limit the application of these methods in practical scenarios.

% However, these algorithms decompose the original problem into multiple MCF subproblems hierarchically and greedily; when there are more than two OCSes in the cluster, such a greedy strategy cannot guarantee the minimum rewiring ratios, and furthermore, a larger number of OCSes in the cluster tends to lead to both suboptimal rewiring ratios and high solving overhead.

In this paper, we first propose a bidirectional model to address the deficiencies of previous MCF-based algorithm modeling, then an instance of the \emph{FastReChain} algorithm is presented based on this bidirectional model. In our evaluation, under low network port utilization conditions (typical for networks where logical topologies are uniformly scaled from traffic patterns, accompanied by high traffic bursts and hotspots), our algorithm achieved improvements of up to 99.9\% in solving time and up to 98.9\% in rewiring ratio over the MCF-based ToE algorithms. When network port utilization was set at 100\%, our algorithm still achieved improvements of up to 97.4\% in solving time and 56.2\% in rewiring ratio over the MCF-based ToE algorithms. These results demonstrate the superior performance of our algorithm, whose feasibility is guaranteed by our modeling design and solid theoretical analysis.

% These results demonstrate the superior performance of our algorithm, which is supported by sophisticated algorithmic ideas and a solid theoretical foundation.

% \zzh{In this paper, we propose an algorithm that does not rely on reducing to the subproblem with two OCSes. In contrast, our method takes a completely different approach: instead of reducing the problem to a specific scale to make it solvable via the off-the-shelf MCF model, it directly tackles the methodology of constructing the OCS port matching scheme itself, thus achieving excellent scalability and outperforming existing algorithms in medium- to large-scale networks. In our evaluation, under conditions of low network port utilization (which is typical for networks with logical topologies uniformly scaled from traffic patterns, high traffic bursts, and hotspots), our algorithm achieved a solving time of up to 0.2\% and a rewiring ratio of up to 6\% of those of the MCF algorithm. When network port utilization is set to 100\% (which is achieved by filling up the logical topology), our algorithm still managed to achieve up to 6\% of the solving time and up to 55\% of the rewiring ratio of the MCF algorithm, demonstrating superior performance based on delicate algorithmic ideas and a robust theoretical foundation.
% }
% \input{[OLD]/Introduction}
%-------------------------------------------------------------------------------

%-------------------------------------------------------------------------------
\section{Background}\label{motiv}
\subsection{Basic Model}\label{basicmodel}
We begin by formulating an ILP model to characterize the fundamental concepts of ToE algorithms.

% Given a predefined logical topology, the OCS enables programmable bandwidth allocation between directly connected devices via ToE, making ToE algorithm design a critical pillar of OCS-based cluster design. 

% An OCS forwards optical signals natively in the optical domain without packet-level processing, such that establishing an optical circuit between two directly connected devices through the OCS is equivalent to provisioning a dedicated direct link between the switch pair. 

In this work, we focus on a canonical two-layer architecture, where a single tier of Top-of-Rack (ToR) switches is interconnected through a single tier of OCSes. For ease of description, \textbf{we define a \emph{link} as the entirety of all \emph{``physical links''} between a ToR switch and an OCS}; each link has a capacity, and the condition for carrying connections is that the number of connections does not exceed the capacity of the link. We denote the link capacity matrix between OCSes and ToR switches as \( C\in\mathbb{Z}^{n\times m}_{\ge 0} \), where \( C_{i,j} \) represents the capacity of the link between the \( i \)-th OCS and the \( j \)-th ToR switch, \( n \) is the total number of OCSes, and \( m \) is the total number of ToR switches in the cluster.

Formally, we represent the logical topology using a non-negative integer matrix \( D\in\mathbb{Z}^{m\times m}_{\ge0} \), where entry \( D_{j,k} \) denotes the number of dedicated connections between the \( j \)-th and \( k \)-th ToR switch, subject to the symmetry constraint \( D_{j,k}=D_{k,j} \), with \( m \) being the total number of ToR switches in the cluster. We use \( X\in\mathbb{Z}^{n\times m\times m}_{\ge0} \) to denote the current OCS configuration, where \( X_{i,j,k} \) represents the number of connections established between the \( j \)-th and \( k \)-th ToR switches via the \( i \)-th OCS, while also satisfying the symmetry constraint \( X_{i,j,k}=X_{i,k,j} \). Similarly, we denote the next OCS configuration to be obtained via ToE as \( Y \), where \( Y_{i,j,k} \) represents the number of connections established between the \( j \)-th and \( k \)-th ToR switches via the \( i \)-th OCS. To satisfy the layer 2 (L2) symmetry constraint, the OCS must establish a bidirectional link for directly connected devices, i.e., \( Y_{i,j,k}=Y_{i,k,j} \). A simple example of the structure of the model is illustrated in Fig.~\ref{fig1}.

A configuration \( Y \) is feasible in the network if it satisfies the following constraints:
\begin{subequations}
\begin{align}
\forall i,j,&\sum_{k=0}^{m-1}Y_{i,j,k}\le C_{i,j}\\
\forall i,j,&\sum_{k=0}^{m-1}Y_{i,k,j}\le C_{i,j}
\end{align}
\end{subequations}

Moreover, the condition for \( Y \) to meet the logical topology requirement \( D \) can be formulated as:
\begin{equation}
\forall j,k,\sum_{i=0}^{n-1}Y_{i,j,k}\ge D_{j,k}
\end{equation}

To reduce the network convergence time, we focus on minimizing the \textit{number of rewirings (NR)}, which is formally defined as:
\begin{equation}\label{eq:nr}
\textrm{NR}(X,Y)=\sum_{i=0}^{n-1}\sum_{j=0}^{m-1}\sum_{k=0}^{m-1}\left|Y_{i,j,k}-X_{i,j,k}\right|
\end{equation}

It is important to note that the model we propose differs from previous works in two key aspects.

First, in previous studies \cite{2021Gemini,zhao2019minimal,han2024lumoscore,zhang2023reducing}, all ToR switches are uniformly connected to all OCSes, meaning that the number of links between each ToR switch and all OCSes is a fixed value; however, our model does not rely on this assumption.

Second, previous models \cite{zhao2019minimal} assume that OCSes can establish arbitrary connections between ToR switches. This assumption is impractical in real-world clusters, because ToE in actual deployments \cite{poutievski2022jupiter,zu2024resiliency} must satisfy a symmetry constraint: ToR switches need to establish bidirectional connections through OCSes. If this constraint is not met, some L2 protocols will fail to function properly \cite{han2024lumoscore,EthernetAddressResolution1982}.

We define our modeling approach incorporating this symmetry constraint as the \textbf{bidirectional} model, while referring to the models established in previous works as the \textbf{traditional} model, where the symmetry constraint is not imposed, i.e., \(D_{j,k} \ne D_{j,k}\), \( X_{i,j,k} \ne X_{i,k,j} \), and \( Y_{i,j,k} \ne Y_{i,k,j} \) may occur. Additionally, each ToR switch can be abstracted as a pair of input switch and output switch, with the input switches containing up-links and the output switches containing down-links. A simple example of the traditional model is illustrated in Fig.~\ref{fig2}. Although previous studies  can greedily reduce the original problem, along with the logical topology, to ones based on the traditional model; solve ToE using algorithms designed for the traditional model; and then, map the solutions back to the original problem \cite{han2024lumoscore}, this greedy strategy itself incurs certain losses in minimizing the objective function formulated in Eq.~\eqref{eq:nr}. Thus, unlike the strategies adopted in previous works, \textbf{we attempt to design ToE algorithms directly based on the bidirectional model in this paper.}

\begin{figure}[!t]
\centerline{\includegraphics{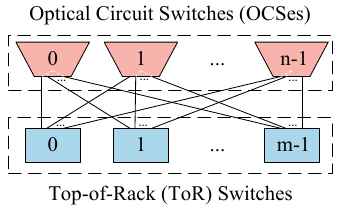}}
\caption{A simple example of the structure of the bidirectional model.}
\label{fig1}
\end{figure}

\begin{figure}
\centerline{\includegraphics{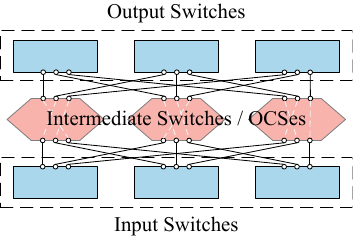}}
\caption{A simple example of the traditional model.}
\label{fig2}
\end{figure}

For convenience, we introduce some notations in the bidirectional model:
\begin{itemize}
    \item $T_i,i=0,1,\ldots,n-1$ represents the $i$-th OCS;
    \item $L_j,j=0,1,\ldots,m-1$ represents the $j$-th ToR switch;
    \item $(T_i,L_j)$ represents the link between $T_i$ and $L_j$; %可能有多根link，有歧义
    \item $(T_i,L_j,L_k)$ represents a connection between $L_j$ and $L_k$ established by the forwarding of $T_i$;
    \item $E(X)_{j,k}$ represents the number of connections between $L_j$ and $L_k$ in scheme $X$, i.e. $E(X)_{j,k}=\sum_iX_{i,j,k}$.
\end{itemize}

All notations used in the bidirectional model are listed in Table~\ref{tab1}.

\begin{table}[!t]
\caption{Notations in the bidirectional model.}\label{tab1}
\begin{center}
\setlength{\tabcolsep}{1.0mm}
\begin{tabular}{c|c}
\toprule
$m$&The number of ToR switches\\
$n$&The number of OCSes\\
\midrule
$L_j$&The $j$-th ToR switch, $0\le j\le m-1$\\
$T_i$&the $i$-th OCS, $0\le i\le n-1$\\
\midrule
$(T_i,L_j)$&The link between $T_i$ and $L_j$\\
\midrule
\multirow{2}{*}{$(T_i,L_j,L_k)$}&A connection between $L_j$ and $L_k$ that\\
&passes through $T_i$\\
\midrule
$C_{i,j}$&The link capacity between $T_i$ and $L_j$\\
\midrule
\multirow{2}{*}{$D_{j,k}$}&The number of connections required\\
&between $L_j$ and $L_k$\\
\midrule
\multirow{4}{*}{$X_{i,j,k},Y_{i,j,k}$}&The number of connections between\\
&$L_j$ and $L_k$ that pass through $T_i$\\
&in the original and the new network\\
&routing scheme, respectively\\
\midrule
\multirow{2}{*}{$E(X)_{j,k}$}&The number of connections between\\
&$L_j$ and $L_k$ in scheme $X$\\
\bottomrule
\end{tabular}
\end{center}
\end{table}

\noindent \textbf{Remark.} Our proposed model is not limited to OCSes, but can be extended to encompass any scenario with essentially the same mathematical model (both the bidirectional and the traditional model). We adopt a flat topology for the purpose of illustrating our approach, owing to its simplicity and ease of comprehension. Nevertheless, the flat models serve as the core foundation for supporting various network topologies and can be extended to multi-layer Clos topologies in a relatively trivial way. For example, in a 2-layer Clos topology, we can adopt a 2-step solution method that solves the outer layer first and then the inner layer, where exactly the same algorithm is called in both steps.

\subsection{Analyzing ToE Optimality}

OCS is a device that directly routes optical signals without processing data packets. Therefore, establishing an optical circuit-level path for directly connected devices through OCS is equivalent to adding direct links between these devices. When the logical topology, i.e., the bandwidth requirements between directly connected devices, is determined, OCS can provide programmable bandwidth for directly connected devices through ToE. Therefore, the design of the ToE algorithm is a key aspect of OCS-based cluster design. We have the following metrics that can be used to evaluate the solution quality of the ToE algorithm.

% Several ToE algorithms have been proposed in the existing literature. However, these algorithms all have their own shortcomings, resulting in a lack of suitable algorithms in certain scenarios, and there is still room for improvement in the overall capabilities of the algorithms. 

\textbf{Solving Overhead: }The solution time of ToE may have an impact on cluster performance. For OCS-based GPU clusters such as TopoOpt \cite{wang2022topoopt} and TPUv4 \cite{zu2024resiliency}, the logical topology may need task-level updates to meet bandwidth requirements, and correspondingly, the ToE algorithm is also performed at task-level. According to public ML traces \cite{hu2024characterizationlargelanguagemodel,hu2021characterization}, in actual GPU clusters, although ML training tasks themselves may last for hours or even longer, users may terminate tasks early for reasons such as parameter tuning, which results in an average task submission interval of less than five minutes \cite{heliosdata2021,AcmeTrace2024,han2024lumoscore}. This means that if we adopt simple ILP-based algorithms, the solution overhead may become a performance bottleneck when dealing with large-scale clusters. To evaluate the impact of solving overhead on ML cluster efficiency, we conducted simulation experiments under two different cluster scales using public datasets \cite{heliosdata2021} and the RapidAISim simulator \cite{han2024lumoscore}. The clusters adopted a leaf-spine-OCS architecture \cite{GoogleA3Supercomputers}, where the hierarchical structure of leaf and spine switches forms a Pod, and multiple Pods are interconnected via OCSes. Each pod is equipped with 128 GPUs; for a cluster consisting of 64 pods, the ILP-based strategy increases the average job running time (JRT, defined as the time from job initiation to job completion) by 6.7\%, with a maximum increase of 45.46\% in JRT compared with MCF-based algorithms. This performance degradation is primarily attributed to the excessively high solving latency of the ILP-based ToE problems. In addition, the job completion time (JCT, defined as the time from job submission to job completion) experiences a maximum increase of 38.16\%, which can be explained by modeling the cluster as a queueing system. According to queuing theory \cite{kingman1962some}, an increase in the average JRT is likely to induce a substantial increase in job queuing time, which further underscores the critical importance of minimizing the solving overhead of ToE algorithms.

% Each Pod contains 128 GPUs; when the cluster includes 128 Pods, the ILP-based strategy increases the average task runtime by xxx\%, mainly due to the excessively long ToE solving latency of the ILP-based strategy. Additionally, the average task waiting time increases by xxx\%, which can be attributed to viewing the cluster as a queuing system. According to queuing theory \cite{kingman1962some}, an increase in average task runtime may lead to a significant rise in task waiting time, further highlighting the importance of solving overhead.

\textbf{Rewiring Ratio: }OCS reconfiguration causes re-pairing of optical modules in directly connected devices, which leads to BGP re-convergence in actual production environments. As shown in Table~\ref{tbl:bgp_conv_ratio}, we evaluated BGP convergence speed under different rewiring ratios, demonstrating that frequent OCS reconfiguration can degrade task communication efficiency. As shown in Fig.~\ref{reconfiguration_cost}, we conduct real-world experiments to further validate the impact of OCS reconfiguration on task performance. Specifically, we conducted experiments by running a LLaMA-7B workload on a cluster containing 16 GPUs and 2 servers, where half of the links between two servers were added or removed; we then recorded the epoch-level training latency. We observed a significant latency surge (from 1.105 s to 2.499 s) in the corresponding training epoch during OCS reconfiguration, which is attributable to BGP reconvergence and reconfiguration-induced packet loss. While sophisticated routing strategies (e.g., diverting in-flight traffic to alternate paths prior to OCS reconfiguration to avoid packet loss) can mitigate the performance degradation induced by OCS reconfiguration, such strategies incur non-trivial additional development and implementation overhead, and still cannot fully eliminate the adverse impact of OCS reconfiguration on ML training throughput. Table~\ref{tbl:reconfig_freq} further validates that minimizing the rewiring ratio reduces the adverse impact of OCS reconfiguration on ML training throughput.

\begin{table}[!t]
	\caption{The Impact of OCS Reconfiguration on BGP Convergence Time Under Different Network Scales}\label{tbl:bgp_conv_ratio}
	\centering 
 \setlength{\tabcolsep}{1.0mm}
	\begin{tabular}{c|ccccc} 
	    \toprule  
	    Rewiring Ratio & 12.5\% & 25\% & 50\% \\
	    \midrule  
		\emph{Avg. BGP convergence time (s)}& 2.19 & 2.38 & 2.6\\
		\bottomrule 
	\end{tabular}
	
\end{table}

\begin{table}[!t]
	\caption{The impact of different reconfiguration frequencies for a llama(7b) task, \emph{Nan} means no recofiguration during training.}\label{tbl:reconfig_freq}
	\centering 
 \setlength{\tabcolsep}{1.0mm}
	\begin{tabular}{c|ccccc} 
	    \toprule  
	    Reconfiguration Interval (s)&30&60&90&\emph{Nan}\\
	    \midrule
		\emph{Avg. overhead per epoch (ms)}& 1175.4 & 1112.8 & 1103.2&1103.0\\
        
		\bottomrule
	\end{tabular}
	
\end{table}

\begin{figure}[!t]
    % \centering）
    \begin{minipage}[b]{0.48\linewidth}
        \centering
        \includegraphics[width=\linewidth]{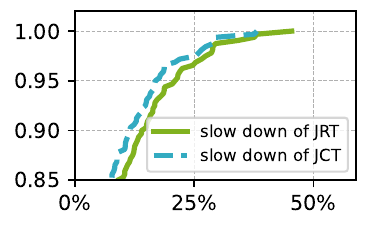}
        \caption{The CDFs of the slowdown ratios for both JRT and JCT under ILP-based ToE algorithms demonstrate that minimizing the solving overhead is critical for OCS-based ML clusters.}
        \label{solving_overhead}
    \end{minipage}\hfill
    \begin{minipage}[b]{0.48\linewidth}
        \centering
        \includegraphics[width=\linewidth]{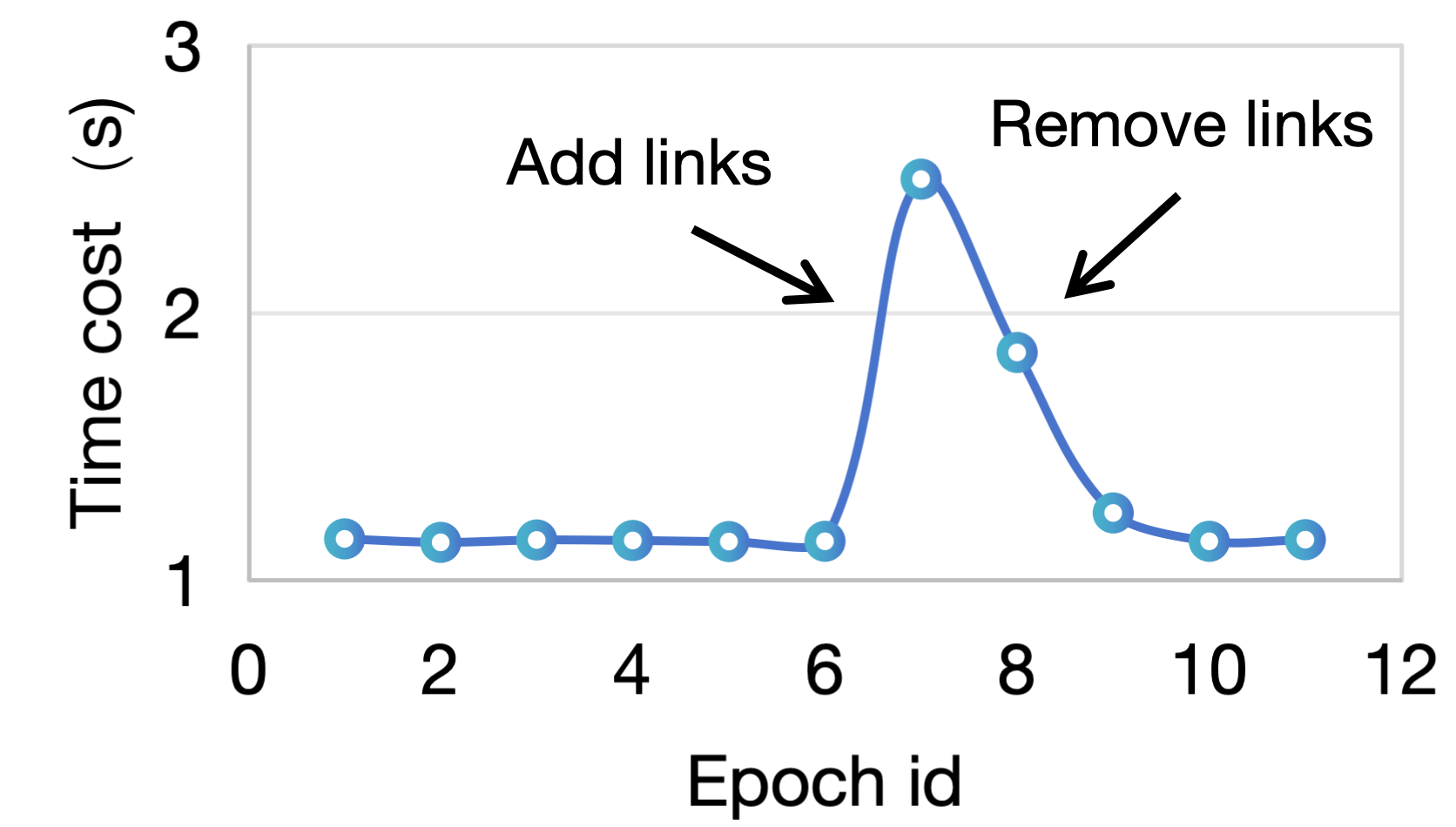}
        \caption{OCS reconfiguration can cause packet loss, resulting in a significant reduction in ML training throughput for a specific time window, making it essential to minimize the rewiring ratio.}
        \label{reconfiguration_cost}
    \end{minipage}
    % \caption{Solving overhead and Rewiring Ratio are important when performing ToE}
\end{figure}

In this paper, we aim to develop a ToE algorithm that achieves both an optimal Rewiring Ratio and minimal solving overhead.

\subsection{Current ToE Algorithm Designs}\label{curdesign}
Several ToE algorithms have been proposed in the existing literature. However, these algorithms all have their own shortcomings, resulting in a lack of suitable algorithms in certain scenarios, and there is still room for improvement in the overall capabilities of the algorithms. 

The most straightforward approach involves solving the ILP model defined in Section~\ref{basicmodel} using optimization solvers such as Gurobi \cite{gurobi}, which has been implemented in deployments such as Jupiter Evolving \cite{poutievski2022jupiter} and TPUv4 \cite{zu2024resiliency}. While this method theoretically yields a strictly optimal solution in terms of Rewiring Ratio, it becomes excessively time-consuming (taking hours or even days) in large-scale clusters, potentially degrading performance for OCS-based GPU clusters requiring task-level ToE execution.

To avoid the high solving overhead caused by ILP-based algorithms, Google and Zhao et al. \cite{zhao2019minimal,han2024lumoscore,zhang2023reducing} proposed solving the ToE problem by building MCF models. While these methods offer polynomial-time complexity, they still exhibit key limitations:

\textcircled{1} Regarding the rewiring ratio, these approaches adopt a divide-and-conquer strategy, splitting the original problem into multiple MCF subproblems. When the cluster contains more than two OCS devices, these methods \textbf{cannot guarantee} \emph{minimum rewiring ratios}, and their effectiveness in minimizing the rewiring ratio diminishes as the number of OCS devices increases. An example of a suboptimal solution produced by the divide-and-conquer strategy is attached in Appendix~\ref{biploss}. Additionally, these algorithms are designed based on the traditional model. Although strategies from works such as \cite{han2024lumoscore} and Sections~\ref{existrc}~and~\ref{modeladap} ensure that these MCF-based methods can find feasible solutions, this approach further increases the rewiring ratio. We further demonstrate the losses incurred by this greedy strategy in optimizing the objective function in Fig.\ref{figMapping}.

\begin{figure}[!t]
    \centering

    \includegraphics[width=0.8\linewidth]{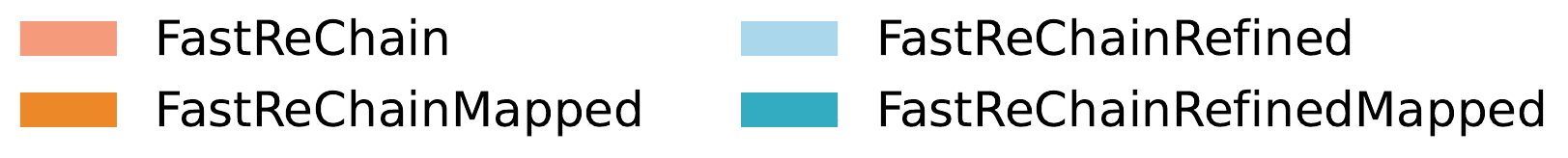}

    \begin{tabular}{@{}c@{}c@{}c@{}}
    
    \subfloat[\(n=128,c=4\)]{
    \includegraphics[width=0.3\linewidth]{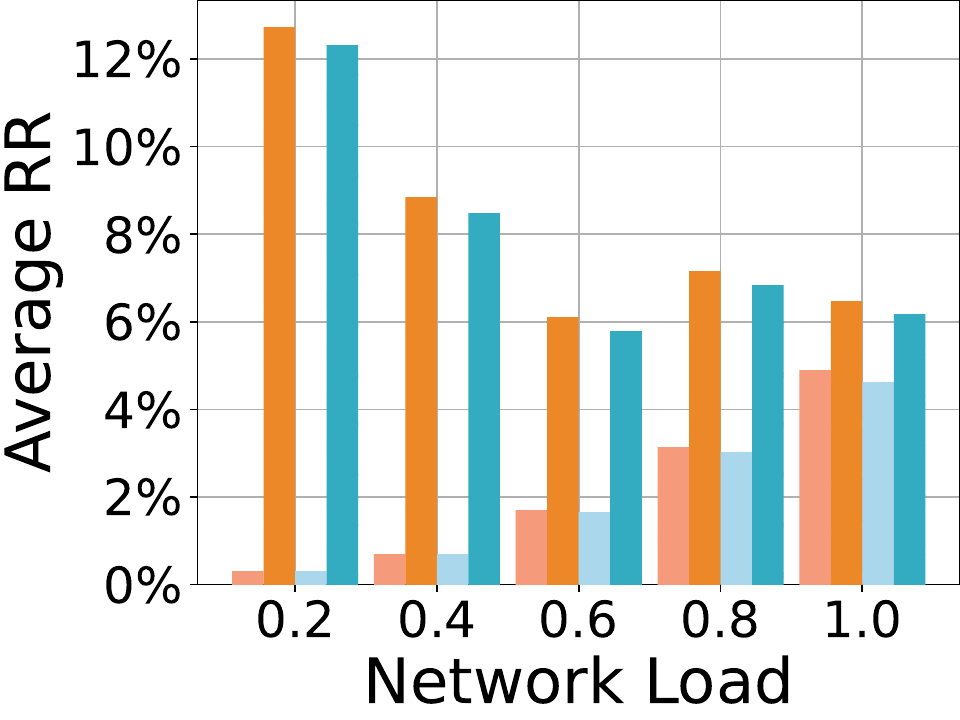}
    }
    &
    \subfloat[\(n=128,c=8\)]{
    \includegraphics[width=0.3\linewidth]{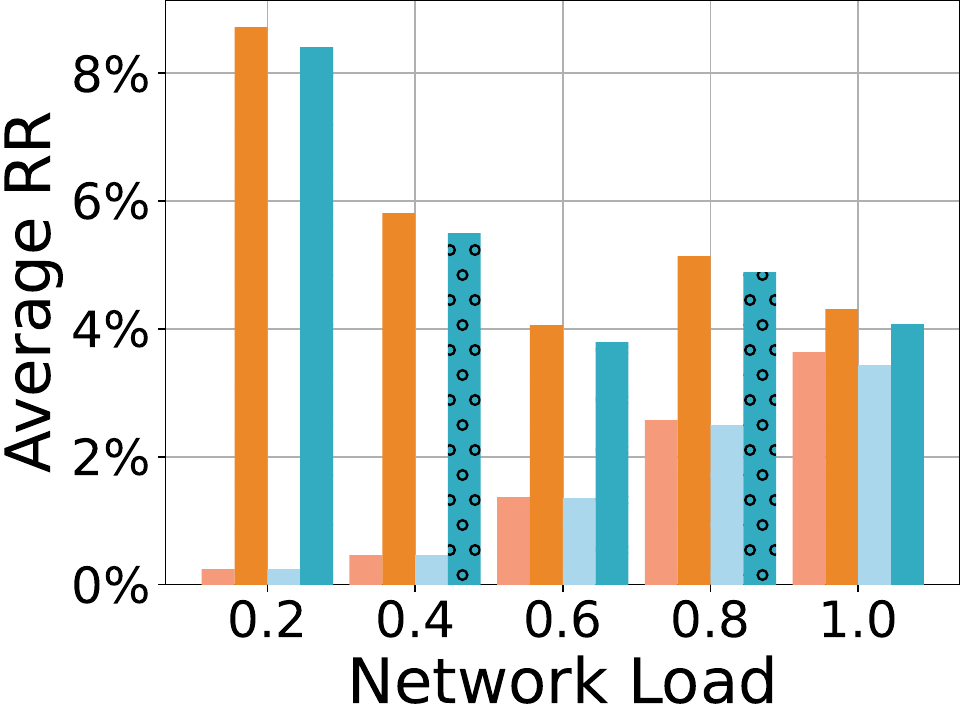}
    }
    &
    \subfloat[\(n=128,c=16\)]{
    \includegraphics[width=0.3\linewidth]{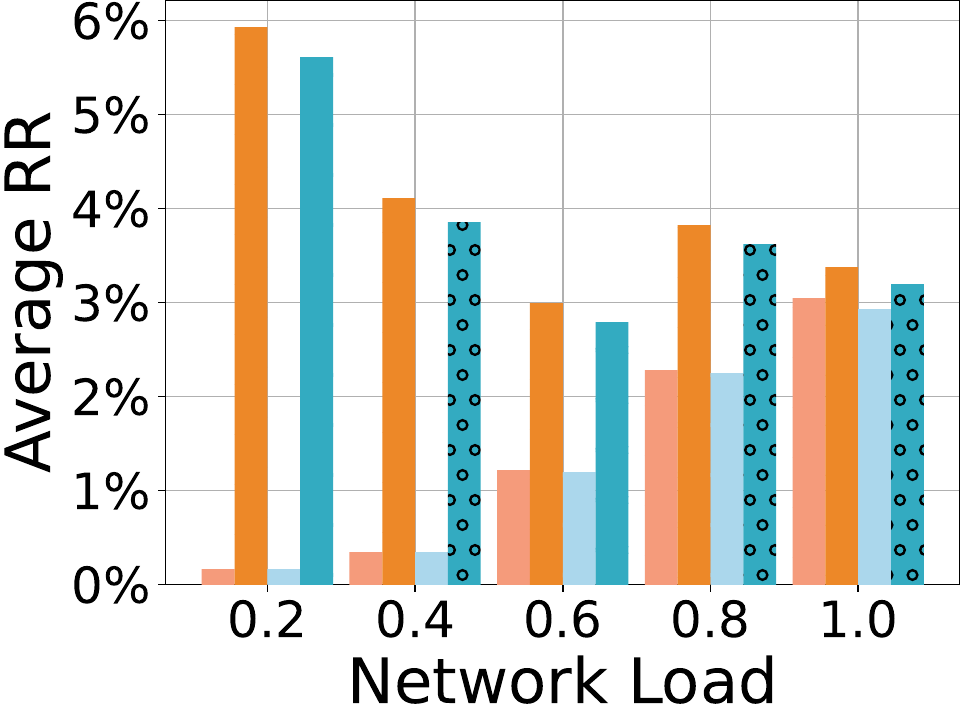}
    }

    \\
    
    \subfloat[\(n=256,c=4\)]{
    \includegraphics[width=0.3\linewidth]{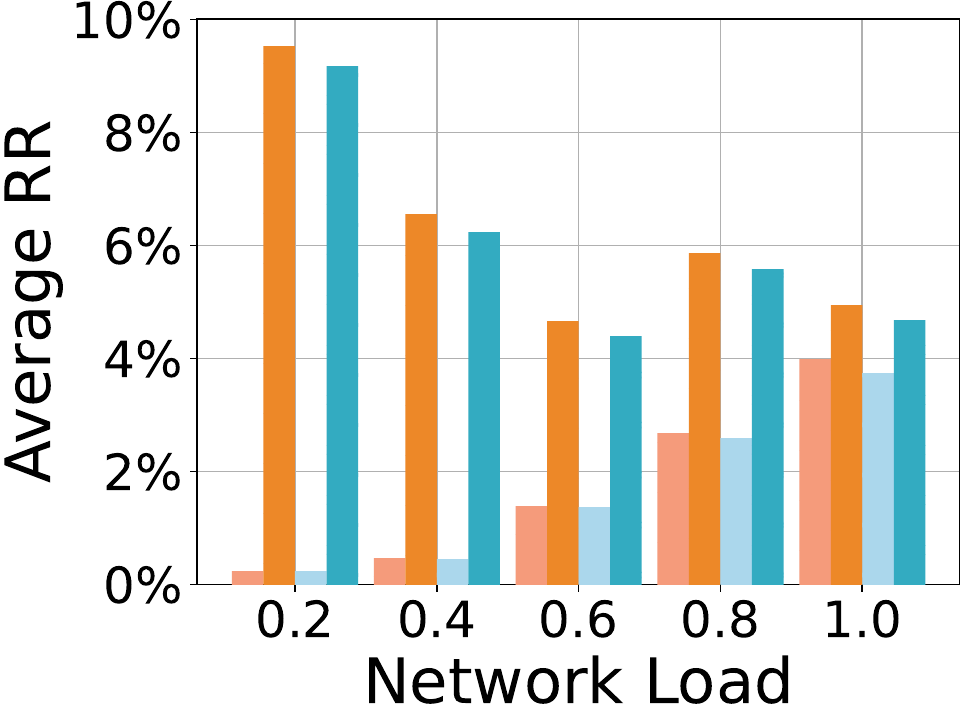}
    }
    &
    \subfloat[\(n=256,c=8\)]{
    \includegraphics[width=0.3\linewidth]{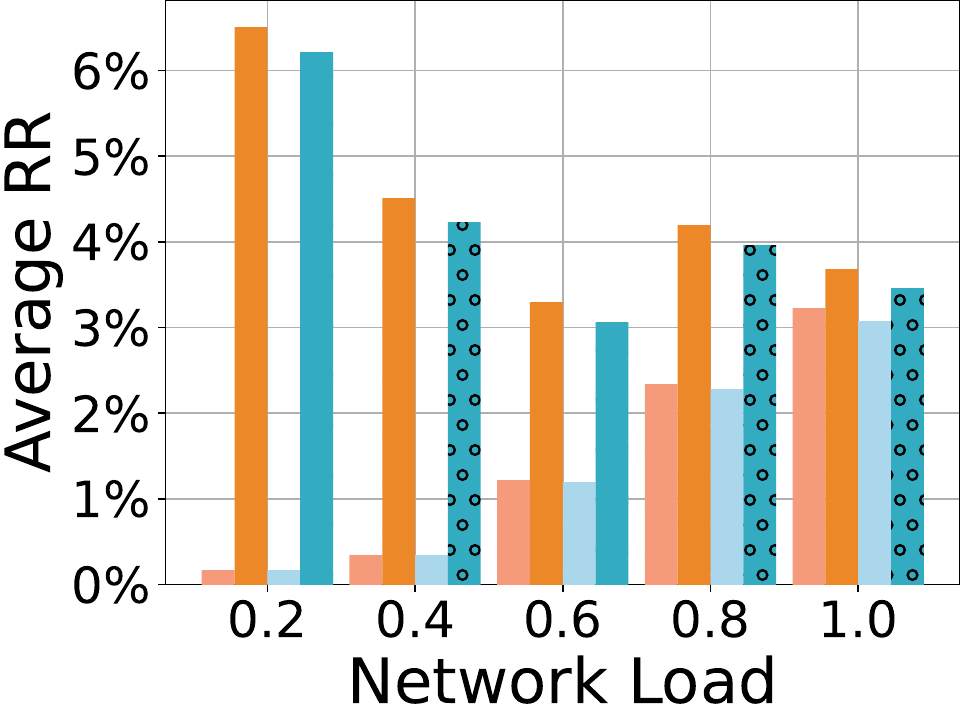}
    }
    &
    \subfloat[\(n=256,c=16\)]{
    \includegraphics[width=0.3\linewidth]{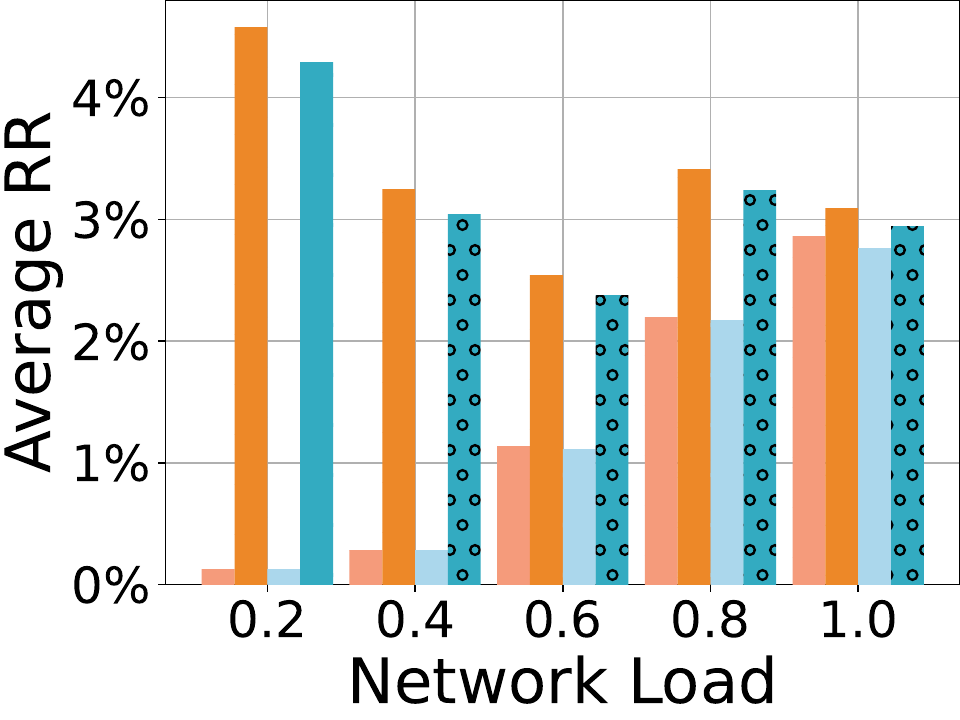}
    }

    \\
    
    \subfloat[\(n=384,c=4\)]{
    \includegraphics[width=0.3\linewidth]{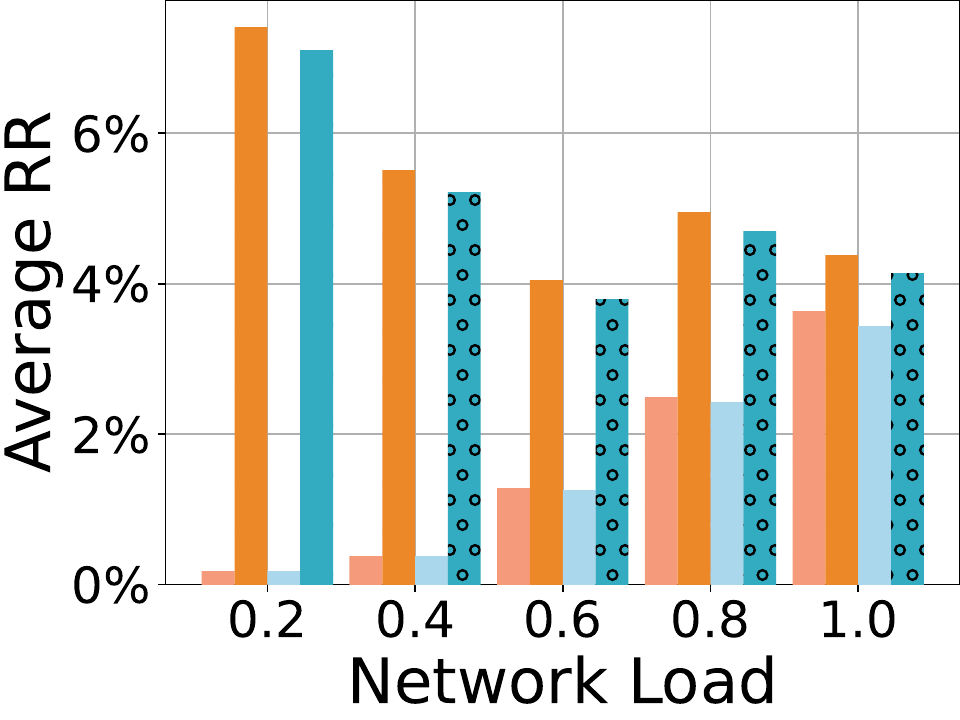}
    }
    &
    \subfloat[\(n=384,c=8\)]{
    \includegraphics[width=0.3\linewidth]{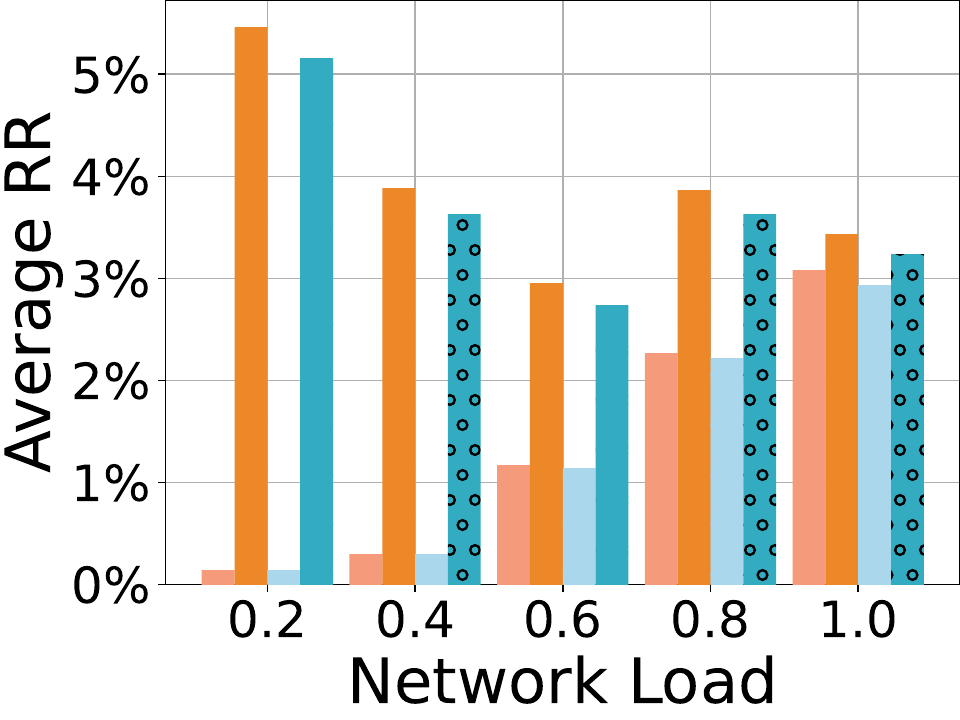}
    }
    &
    \subfloat[\(n=384,c=16\)]{
    \includegraphics[width=0.3\linewidth]{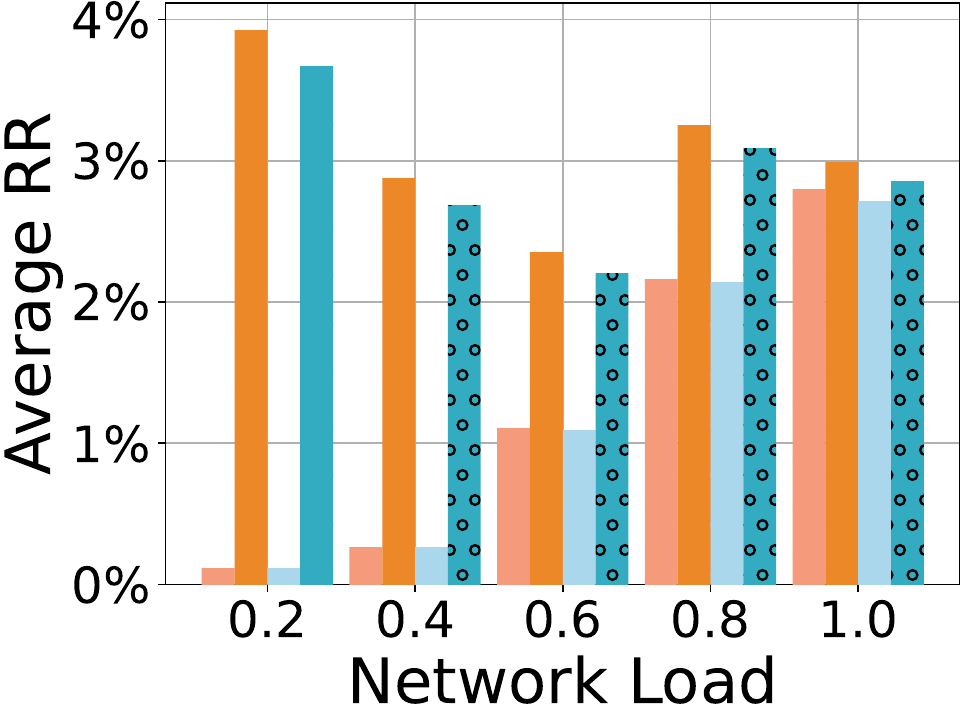}
    }

    \end{tabular}
    
    \caption{Mapping loss quantified in the bidirectional model. Instances labeled with the ``mapped'' suffix are versions that use a bidirectional model to traditional model conversion.}\label{figMapping}
\end{figure}

\textcircled{2} In large-scale networks, the solving overhead of these MCF-based methods requires further optimization. These methods construct MCF models using cluster information including switch and OCS device counts, but fail to account for the magnitude of logical topology changes. This means even minor logical topology alterations require solution time on the same order of magnitude as significant changes. Consequently, MCF algorithm solution times for ToE problems in large-scale networks remain suboptimal and need optimization. 

% As shown in Fig.\ref{}, larger network scales correspond to smaller traffic variation magnitudes under identical sampling intervals, resulting in smaller logical topology change magnitudes. 

If we start from the assumption that the traffic changes are not significant, then a straightforward idea is to do modifications directly in the original port matching scheme. We are primarily concerned with the additions compared to the original logical topology, so we begin by considering how to add these circuit paths one by one, thus forming the starting point of our algorithm.

\section{The Replacement Chain Concept}\label{rcconcept}
%-------------------------------------------------------------------------------

%-----------------------------------
\subsection{Replacement Chain}\label{formrc}
%-----------------------------------
First, we introduce some basic definitions about replacement chain and the algorithm.

\begin{definition}
We call the connections between $L_j$ and $L_k$ \textbf{redundant} in a scheme $X$, if and only if $E(X)_{j,k}>D_{j,k}$.
\end{definition}

\begin{definition}
A link $(T_i,L_j)$ is \textbf{explicitly available} in a scheme $X$, if and only if the link is not fully used, i.e. $\sum_kX_{i,j,k}<C_{i,j}$;

A link $(T_i,L_j)$ is \textbf{implicitly available} in a scheme $X$, if and only if the link is fully used, i.e. $\sum_kX_{i,j,k}=C_{i,j}$, and there exist some $L_k$ such that $X_{i,j,k}>0$ and the connections between $L_j$ and $L_k$ are redundant;

A link is \textbf{available} if it is explicitly or implicitly available.
\end{definition}

When a link is implicitly available, we can make it explicitly available by removing one of the corresponding redundant connections.

\begin{definition}
We define two types of atomic modifications as follows:
\begin{itemize}
\item $\mathrm{Add}(i,j,k)$: add one connection between $L_j$ and $L_k$ established through of $T_i$. This is equivalent to adding $1$ to $X_{i,j,k}$ and $X_{i,k,j}$ if the operation is carried out in scheme $X$.
\item $\mathrm{Remove}(i,j,k)$: Remove one connection between $L_j$ and $L_k$ established through $T_i$. This is equivalent to reducing $1$ from $X_{i,j,k}$ and $X_{i,k,j}$ if the operation is carried out in scheme $X$.
\end{itemize}
\end{definition}

In order to form a replacement chain, we start from considering adding a new connection without affecting the rest of the logical topology. We pick a pair of ToR switches $L_{j_0}$ and $L_{k_0}$ such that $E(X)_{j_{0},k_{0}}<D_{j_{0},k_{0}}$, and consider how to schedule a connection between $L_{j_0}$ and $L_{k_0}$ in scheme $X$:
\begin{itemize}
    \item If for all $i$, $(T_i,L_{j_0})$ and $(T_i,L_{k_0})$ are all unavailable, then the connection cannot be added without reducing the number of non-redundant connections;
    \item  Otherwise, we choose an OCS $T_{i_0}$ and try to add $(T_{i_0},L_{j_0},L_{k_0})$:
    \begin{enumerate}
        \item If $(T_{i_0},L_{j_0})$ and $(T_{i_0},L_{k_0})$ are all available, we can comfortably add the connection without reducing the number of non-redundant connections;
        \item If one of them, say $(T_{i_0},L_{j_0})$ is available and $(T_{i_0},L_{k_0})$ is not, then we have to permute some connection $(T_{i_0},L_{j_1},L_{k_0})$ to make $(T_{i_0},L_{k_0})$ available;
        \item If both links are not available, then we have to permute two connections and reduce to adding two connections.
    \end{enumerate}
\end{itemize}
In the 2\textsuperscript{nd} case, the problem then reduces to adding back a connection between $L_{j_1}$ and $L_{k_0}$, which is exactly of the same form and can therefore be solved recursively. Nonetheless, the 3\textsuperscript{rd} case is too complex and can potentially lead to an extremely large amount of computation. Therefore, it is not adopted in our algorithm and we simply skip this OCS when we encounter this case.

The recursion ends when both links are available. Note that if we need to use an implicitly available connection at any point in the process, we need to remove one of the corresponding redundant connections simultaneously.

As a result, we introduce the definition of replacement chain.

\begin{definition}
A replacement chain of length $l$ is defined as the following sequence of atomic modifications:
\begin{equation*}
\begin{gathered}
\mathrm{Remove}(i_0,j_1,k_1),\mathrm{Add}(i_0,j_0,k_0),\\
\mathrm{Remove}(i_1,j_2,k_2),\mathrm{Add}(i_1,j_1,k_1),\\
\mathrm{Remove}(i_2,j_3,k_3),\mathrm{Add}(i_2,j_2,k_2),\\
\vdots\\
\mathrm{Remove}(i_{l-1},j_{l},k_{l}),\mathrm{Add}(i_{l-1},j_{l-1},k_{l-1}),\\
\mathrm{Add}(i_{l},j_{l},k_{l})
\end{gathered}
\end{equation*}
where
\begin{equation*}
\forall t,(j_t=j_{t+1}\land k_t\ne k_{t+1})\vee(j_t\ne j_{t+1}\land k_t=k_{t+1})
\end{equation*}
\end{definition}

We can see that the resulting modification scheme naturally forms a replacement chain as defined above.

As a result of applying the replacement chain to a scheme $Y$, $E(Y)_{j_0,k_0}$ increases by $1$, while the other elements in $E(Y)$ remain the same. Such a sequence is called a replacement chain because it tries to add $(T_{i_0},L_{j_0},L_{k_0})$ to the scheme and permutes $(T_{i_0},L_{j_1},L_{k_1})$, then it continues trying to add $(T_{i_1},L_{j_1},L_{k_1})$ to the scheme and permutes $(T_{i_1},L_{j_2},L_{k_2})$, etc. In fact, Paull has published a similar concept, but limited to replacing between two intermediate switches in symmetric Clos networks \cite{paull1962reswitching}. Ohta and Ueda also discussed this concept \cite{ohta1987rearrangement}.

\subsection{Restricting to Two OCSes}\label{resttwo}

There are many interesting properties when restricting to 2 OCSes, which serves as the foundation for the cases of multiple OCSes.

In the following text, we represent these 2 OCSes as \(A\) and \(B\). For ease of comprehension, we first describe the situation in the traditional model \textbf{with link capacities equal to 1}. We represent the \(j\)-th input switch abstracted from the \(j\)-th ToR switch as \(U_j\), and the \(k\)-th output switch abstracted from the \(k\)-th ToR switch as \(V_k\). Similarly, the replacement chain limited to two OCSes is in the following form:
\begin{equation*}
\begin{gathered}
\mathrm{Remove}(A,j_0,k_1),\mathrm{Add}(A,j_0,k_0),\\
\mathrm{Remove}(B,j_1,k_1),\mathrm{Add}(B,j_0,k_1),\\
\mathrm{Remove}(A,j_1,k_2),\mathrm{Add}(A,j_1,k_1),\\
\vdots\\
\end{gathered}
\end{equation*}
where \(\mathrm{Remove}(A/B,j,k)\) represents the removal of a connection between \(U_j\) and \(V_k\) established through \(A/B\), and \(\mathrm{Add}(A/B,j,k)\) represents the addition of a connection between \(U_j\) and \(V_k\) established through \(A/B\). Note that we assume that the links between \(A\) and \(V_{k_0}\) and between \(B\) and \(U_{j_0}\) are available. It seems that the replacement chain at this point is essentially a series of existing connections: \((A,U_{j_0},V_{k_1}),(B,U_{j_1},V_{k_1}),(A,U_{j_1},V_{k_2}),(B,U_{j_2},V_{k_2}),\ldots\), and the construction of the replacement chain is completed when the replacement chain extends to an available link. A valid replacement chain can always be found by extending on the series of existing connections because, in the extension process of the replacement chain, it will never visit a visited link; otherwise, the link would have more than one connection, contradicting the fact that the link capacity is equal to 1. We can draw the connections in \(A\) and \(B\) in the same graph, with the connections of \(A\) as solid lines and the connections of $B$ as dashed lines. Then, finding a replacement chain is equivalent to alternating between solid and dashed lines until we reach an available link, as shown in Fig.~\ref{fig4}. 

\begin{figure}[!t]
\centerline{\includegraphics[scale=1.2]{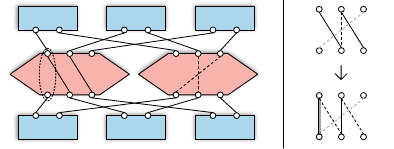}}
\caption{Replacement between two OCSes in a traditional network.}
\label{fig4}
\end{figure}

Therefore, the algorithm in this case is to perform a breadth-first search (BFS) on this graph. Since each link can be visited at most once, the time complexity is linear in the number of input/output switches. Note that in the case where the link capacity is greater than 1, a link may be visited more than once, but it will not contribute to a shorter replacement chain, because BFS always finds the shortest path to the available links.

In the bidirectional model, there is a similar approach, but a link, say $(A,L_j)$, should be allowed to be visited twice, once in the form of $(A,L_j,*)$ and once in the form of $(A,*,L_j)$. This is a corollary of the bipartite graph coloring in \ref{existrc}.

%-----------------------------------
\subsection{Existence of Replacement Chains}\label{existrc}
%-----------------------------------

Our goal is to prove the existence of replacement chains in an extended form of the network models, rather than the basic uniform models, that is, the proportional models, as defined below.

\begin{definition}\label{def5}
A traditional model is defined as \textbf{proportional}, if there exist $W_{L_j}\in\mathbb{Z}_{\ge0}$ for each $U_j$ and $V_j$, and $W_{T_i}\in\mathbb{Z}_{\ge0}$ for each $T_i$, such that $C_{i,j}=W_{T_i}\cdot W_{L_j}$.
\end{definition}
\begin{definition}
A bidirectional model is defined as \textbf{proportional}, if there exist $W_{L_j}\in\mathbb{Z}_{\ge0}$ for each $L_j$ and $W_{T_i}\in\mathbb{Z}_{\ge0}$ for each $T_i$, such that $C_{i,j}={\color{red}2}\cdot W_{T_i}\cdot W_{L_j}$.
\end{definition}
\textit{Proportional} encompasses \textit{uniform} (i.e., all links having the same capacity).

Note that having even link capacities is a key condition for guaranteeing the existence of a solution in the bidirectional model. A similar design is used in \cite{han2024lumoscore}.

We start by proving the theorem in the traditional model, as it serves as the foundation for the proof in the bidirectional model.

\begin{lemma}\label{lem}
In the proportional traditional model, for the scheduling of a single connection between $U_{j_0}$ and $V_{k_0}$, assume that the link between $T_{i_0}$ and $U_{j_0}$ and the link between $T_{i_1}$ and $V_{k_0}$ are not fully used for some $T_{i_0}$ and $T_{i_1}$, then a valid replacement chain can always be found by alternately replacing the connections in $T_{i_0}$ and $T_{i_1}$, and subsequently can always be found by the general method.
\end{lemma}

\begin{proof}
Since $C_{i,j}=W_{T_i}\cdot W_{L_j}$, we can divide each input switch $U_j$ and each link connected to it into $W_{L_j}$ switches and links, and the capacity of each divided link is $W_{T_i}$, where $T_i$ is the OCS to which the link is connected. Do the same for the output switches. This process is trivial, as the established connections can be trivially assigned to the divided switches.

Then, we divide each OCS $T_i$ and each link connected to it into $W_{T_i}$ switches and links, and the capacity of each divided link is $1$. It does not seem trivial to assign the connections to the divided OCSes. Consider the most complex case, where all the links connected to an OCS $T_i$ are fully used. If not all links are fully used, one can trivially add additional connections to make them fully used. Then, assigning the connections to the divided OCSes is equivalent to dividing a $W_{T_i}$-regular bipartite graph into $W_{T_i}$ perfect matches. Since $\forall k>0$, every $k$-regular bipartite graph has a perfect matching \cite{harary1969graph}, such a division always exists.

Now, the network is transformed into a traditional model with link capacities equal to 1, and we have shown that the lemma holds in this scenario in Section~\ref{resttwo}. Therefore, the lemma holds in the proportional traditional model.
\end{proof}

Now, we show the proof of the existence of replacement chains in the bidirectional model via a bipartite graph coloring method to map the bidirectional model to the traditional model.

\begin{theorem}
In the proportional bidirectional model, for scheduling a single connection between $L_{j_0}$ and $L_{k_0}$, a valid modification scheme can be found by finding replacement chains, i.e. there exists at least one valid modification scheme that is in the form of a replacement chain and some removal of redundant connections, if and only if there exists at least one link connected to $L_{j_0}$ and at least one link connected to $L_{k_0}$ that are available.
\end{theorem}

\begin{proof}
If all links connected to $L_{j_0}$ or all links connected to $L_{k_0}$ are not available, then apparently there exist no valid schemes. Otherwise, we can map the network to the traditional model, i.e. divide one half of the ports of a link into input ports and the other half into output ports, while ensuring that all connections are between input and output ports.

We can achieve this by constructing a bipartite graph. For each link $(T_i,L_j)$, we create $C_{i,j}$ vertices, namely $V_{i,j,0},V_{i,j,1},\ldots,V_{i,j,C_{i,j}-1}$. We create two types of edges in the graph:
\begin{enumerate}
    \item Edges $(V_{i,j,2t},V_{i,j,2t+1})$ for $0\le t<\frac{C_{i,j}}{2}$;
    \item For each connection $(T_i,L_j,L_k)$, allocate one vertex from the vertices of $(T_i,L_j)$ and one vertex from the vertices of $(T_i,L_k)$ (each vertex can be allocated at most once), and connect them with an edge.
\end{enumerate}
The graph naturally forms a bipartite graph because every vertex is connected to exactly one edge of type 1 and at most one edge of type 2, so that each cycle in the graph alternates between the two types of edges. We then color the vertices black and white so that every pair of vertices connected by an edge has different colors. We direct all the connections corresponding to type 2 edges from black to white, and divide each link into an up-link and a down-link, where the black vertices represent the up-links and the white vertices represent the down-links. Therefore, each ToR switch is divided into an input switch and an output switch, containing up-links and down-links respectively. Fig.~\ref{fig3} shows an example of bipartite graph coloring. 

\begin{figure}[!t]
\centerline{\includegraphics[scale=0.35]{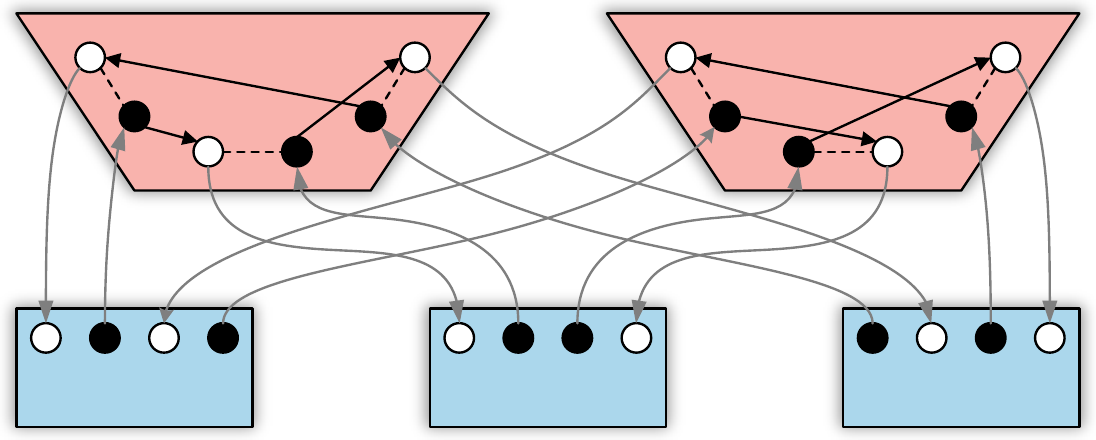}}
\caption{An example of bipartite graph coloring in a bidirectional network ($n=2$, $m=3$, $W_*=1$).}
\label{fig3}
\end{figure}

Note that we need to ensure that there is at least one pair of available links connected to different types of switches divided from $L_{j_0}$ and $L_{k_0}$, respectively. This can always be achieved by carefully choosing the color when coloring the vertices.

The problem now reduces to scheduling a single connection in the proportional traditional model, which is proven in Lemma~\ref{lem}. Therefore, the theorem holds.
\end{proof}

\noindent \textbf{Remark.} While the existence of a solution is strictly guaranteed only in the proportional model, the proposed algorithm can still function correctly in irregularly-structured models, with a probability of not finding a solution when a solution exists theoretically. The probability can be extremely low in real-world situations, and that a tiny number of connections failing to be established is actually insignificant, as long as the connections are added in a reasonably prioritized order, as mentioned in Section~\ref{adap}. Nonetheless, if the 3\textsuperscript{rd} case in Section~\ref{formrc} is adopted, then the algorithm is guaranteed to find feasible solutions when there exist.

% TODO，把证明放在前面一点，有没有办法让人在看定理前就觉得这个想法是对的，能不能把证明拆成小引理，每个小引理都很直观，拼出来就让人感觉是对的。或者能不能给出一个例子让人觉得多个OCS的情况下是可行的，似乎从D_i变到D_{i+1}，这个两个对应OCS的连线对，对比相同的OCS，肯定有些线变有些线不变，根据线的变化有没有可能构造出来一张图，按照图所谓的找一个置换链，就是在图上找欧拉回路，建立这个关系别人就会好懂，正文只保留这个直觉，具体证明放附录里，把这个点立柱，别人就会觉得想什么时候停就什么时候停，实验表明不需要搜那么长时间就能搜到挺好的解了，就可以更好得说明时间开销。

% \noindent \textbf{Remark.} Please note that this is the primary source of approximation loss for our algorithm -- if we adopt the 3\textsuperscript{rd} case, then our algorithm always finds the optimal solution when scheduling a single connection.

%-------------------------------------------------------------------------------
\section{The FastReChain Algorithm}\label{frcalg}
%-------------------------------------------------------------------------------

%-----------------------------------
\subsection{Basic Algorithm}\label{thealg}
%-----------------------------------

We can perform a depth-first search (DFS) to find replacement chains. At each node of the search tree, our goal is to add a connection between $L_j$ and $L_k$. Therefore, we first need to enumerate the OCSes $T_i$ through which the connection will be established. In $T_i$, if the connection can be added directly, the search ends. Otherwise, one of the links from $L_j$ and $L_k$ to $T_i$ needs to be available. Assuming $(T_i,L_j)$ is available, we need to replace a connection to $L_k$ to make $(T_i,L_k)$ available. Therefore, we need to enumerate $j'$, satisfying $X_{i,j',k}>0$, then replace $(T_i,L_{j'},L_k)$, and proceed to the next level of search.

To find the shortest replacement chain, we can use iterative deepening search (IDS), that is, repeatedly run a depth-limited version of DFS with increasing depth limits. We can see that at each node of the search tree, two enumerations are required, namely: 1) enumerate the OCSes $T_i$ through which the connection between $L_j$ and $L_k$ will be established; 2) enumerate the ToR switches $L_{j'}$ to which the replaced connection $(T_i,L_{j'},L_k)$ is connected (assuming that $(T_i,L_j)$ is available and that we need to make $(T_i,L_k)$ available). To ensure the stability of the algorithm performance and the load balancing of the switches, we can perform these two enumerations in random orders.

For detailed pseudocode of the above basic algorithm, see Appendix~\ref{rcids}.

\subsection{Efficiency Analysis}
However, some might argue that this algorithm does not seem to work in practice. The depth of the search, which is also the length of the replacement chain, can be as large as the number of ToR switches. The search breadth, which is equivalent to the number of selectable switches (i.e. switches that can be selected to allow the search to proceed) at each node of the search tree can also be in the order of the number of switches. This yields an extremely large, exponentially increasing amount of computation. How could this algorithm possibly work in practice?

\noindent \textbf{Key observation:} When the network scale is small, the shortest replacement chain could be long, but the search breadth would be small, since the network scale is small. When the network size is large, the shortest replacement chain is likely to be very short because of the large number of possible replacement chains. We also need to take network load into account: When the network load is low, meaning that many or all of the links are available, it is highly possible that there exist replacement chains of very short length. When the network load is high, the number of selectable OCSes and replaceable connections is extremely limited. Therefore, the average time complexity is much better than it seems. In our experiment, even when the network load is 100\% (i.e. all the links are fully used), in a uniform traditional network with 155 input/output switches, 128 OCSes, and link capacities equal to 4, under the traffic traces from Facebook datacenters, the lengths of the shortest replacement chains did not exceed 6. Table~\ref{tab2} shows the distribution of the shortest replacement chain lengths across 143 reconfiguration phases in the experiment.

\begin{table}
\caption{Distribution of the shortest replacement chain lengths in a uniform traditional network.}\label{tab2}
\begin{center}
\small
\begin{tabular}{c|c}
\toprule
Replacement chain length&\emph{Number of schedulings}\\
\midrule
0&1352799\\
1&125233\\
2&35205\\
3&12133\\
4&3061\\
5&342\\
6&20\\
\(\ge\)7&0\\
\bottomrule
\end{tabular}
\end{center}
\end{table}

However, the algorithm still runs pretty slow, especially when the network load is high. Observing and speculating on the algorithm execution process, we found that the bottleneck of the algorithm lies in the enumeration of the switches. When the network load is high, there are very few selectable switches, but we still need to enumerate each switch equally. Especially in the last layer of search, it is common that only one or two switches can complete the construction of the replacement chain, but we need to enumerate all switches, while the search tree is fully expanded exactly at the last layer. Therefore, we need to find a way to filter out selectable switches.

\subsection{Bitset Optimization}\label{bsopt}
We can use bitset to perform the filtering mentioned above. Bitset supports bitwise operations on 64 bits at a time, or even more bits with parallel computing (such as Single Instruction Multiple Data (SIMD) or GPU). Here, we only explain how to use bitsets to filter out OCSes that have available links to some $L_{j_0}$ as an example. The pseudocode is shown in Algorithm~\ref{alg}.

First, we need to check which OCSes have links that are explicitly available. We need to maintain $m$ bitsets $\mathrm{bitsetLT}_0, \mathrm{bitsetLT}_1, \ldots,\mathrm{bitsetLT}_{m-1}$, each of which has a length of $n$. $\mathrm{bitsetLT}_j[i]$ is the $i$-th bit in $\mathrm{bitsetLT}_{j}$, which represents whether the link between $T_i$ and $L_j$ is explicitly available. To help maintain $\mathrm{bitsetLT}$, we also need to maintain a two-dimensional array $\mathrm{countLT}$, where $\mathrm{countLT}_{j,i}$ represents the number of connections that occupy the link between $L_j$ and $T_i$. Then, $\mathrm{countLT}$ and $\mathrm{bitsetLT}$ can be easily updated in $\Theta(1)$ time complexity when a connection is added or removed.

For OCSes with implicitly available links, bitsets cannot be updated in $\Theta(1)$ time directly. We can maintain $m$ bitsets $\mathrm{bitsetLL}_0, \mathrm{bitsetLL}_1, \ldots, \mathrm{bitsetLL}_{m-1}$, each of which has a length of $m$. $\mathrm{bitsetLL}_j[k]$ represents whether the connections between $L_j$ and $L_k$ are redundant. Maintenance of $\mathrm{bitsetLL}$ also requires a two-dimensional array $\mathrm{countLL}$, where $\mathrm{countLL}_{j,k}$ represents the number of connections between $L_j$ and $L_k$. We also need to maintain $\mathrm{bitsetLLT}$, where $\mathrm{bitsetLLT}_{j,k}[i]$ represents whether there exist connections between $L_j$ and $L_k$ that pass through $T_i$. To obtain the bitset of OCSes with implicitly available links to $L_{j_0}$, we need to take the bitwise OR of $\mathrm{bitsetLLT}_{j_0,k}$ for each $k$ that is the index of a 1 in $\mathrm{bitsetLL}_{j_0}$.

In our experiment, the algorithm runs \textasciitilde1000 times faster with bitset optimization. Furthermore, the bitset optimization makes the algorithm's time complexity \textbf{almost independent of the number of network ports or switches}, thus giving the algorithm the potential to scale to a very large scale.

\begin{algorithm}[!t]
    \SetKwProg{Fn}{function}{:}{end}
    \caption{A part of the bitset optimization.}\label{alg}

    \footnotesize
    
    \KwData{$m$, $n$, $C$, $D$, current scheme $X$, $\mathrm{bitsetLT}$, $\mathrm{countLT}$, $\mathrm{bitsetLL}$, $\mathrm{countLL}$, $\mathrm{bitsetLLT}$}
    \BlankLine
    
    \SetKwFunction{AddConnection}{AddConnection}
    \Fn{\AddConnection{$i$, $j$, $k$}} {
        $\mathrm{countLT}_{j,i}\gets \mathrm{countLT}_{j,i}+1$\;
        \textbf{if} $\mathrm{countLT}_{j,i}=C_{i,j}$ \textbf{then}\\
        \quad\ \,$\mathrm{bitsetLT}_{j}[i]\gets 0$\;
        $\mathrm{countLT}_{k,i}\gets \mathrm{countLT}_{k,i}+1$\;
        \textbf{if} $\mathrm{countLT}_{k,i}=C_{i,k}$ \textbf{then}\\
        \quad\ \,$\mathrm{bitsetLT}_{k}[i]\gets 0$\;
        \textbf{if} $\mathrm{countLL}_{j,k}=D_{j,k}$ \textbf{then}\\
        \quad\ \,$\mathrm{bitsetLL}_{j}[k]\gets 1$, $\mathrm{bitsetLL}_{k}[j]\gets 1$\;
        $\mathrm{countLL}_{j,k}\gets \mathrm{countLL}_{j,k}+1$\;
        $\mathrm{countLL}_{k,j}\gets \mathrm{countLL}_{k,j}+1$\;
        \textbf{if} $X_{i,j,k}=0$ \textbf{then}\\
        \quad\ \,$\mathrm{bitsetLLT}_{j,k}[i]\gets 1$, $\mathrm{bitsetLLT}_{k,j}[i]\gets 1$\;
        $X_{i,j,k}\gets X_{i,j,k}+1$, $X_{i,k,j}\gets X_{i,k,j}+1$\;
    }
    \BlankLine

    \SetKwFunction{RemoveConnection}{RemoveConnection}
    \Fn{\RemoveConnection{$i$, $j$, $k$}} {
        \textbf{if} $\mathrm{countLT}_{j,i}=C_{i,j}$ \textbf{then}\\
        \quad\ \,$\mathrm{bitsetLT}_{j}[i]\gets 1$\;
        $\mathrm{countLT}_{j,i}\gets \mathrm{countLT}_{j,i}-1$\;
        \textbf{if} $\mathrm{countLT}_{k,i}=C_{i,k}$ \textbf{then}\\
        \quad\ \,$\mathrm{bitsetLT}_{k}[i]\gets 1$\;
        $\mathrm{countLT}_{k,i}\gets \mathrm{countLT}_{k,i}-1$\;
        $\mathrm{countLL}_{j,k}\gets \mathrm{countLL}_{j,k}-1$\;
        $\mathrm{countLL}_{k,j}\gets \mathrm{countLL}_{k,j}-1$\;
        \textbf{if} $\mathrm{countLL}_{j,k}=D_{j,k}$ \textbf{then}\\
        \quad\ \,$\mathrm{bitsetLL}_{j}[k]\gets 0$, $\mathrm{bitsetLL}_{k}[j]\gets 0$\;
        $X_{i,j,k}\gets X_{i,j,k}-1$, $X_{i,k,j}\gets X_{i,k,j}-1$\;
        \textbf{if} $X_{i,j,k}=0$ \textbf{then}\\
        \quad\ \,$\mathrm{bitsetLLT}_{j,k}[i]\gets 0$, $\mathrm{bitsetLLT}_{k,j}[i]\gets 0$\;
    }
    \BlankLine

    \SetKwFunction{GetBitsetOfSelectableTs}{GetBitsetOfSelectableTs}
    \Fn{\GetBitsetOfSelectableTs{$j$}} {
        $\mathrm{bitset}\gets \mathrm{bitsetLT}_{j}$\;
        \ForEach{$k$ satisfying $\mathrm{bitsetLL}_j[k]=1$} {
            $\mathrm{bitset}\gets\mathrm{bitset}\ \vert\ \mathrm{bitsetLLT}_{j,k}$\;
        }
        \KwRet{$\mathrm{bitset}$}\;
    }
\end{algorithm}

%-----------------------------------
\subsection{Rough Bound on the Length of the Replacement Chain}\label{limittwo}
%-----------------------------------

To illustrate the low average time complexity of our algorithm, we show that it imposes a rough upper bound on the length of the \emph{shortest} replacement chain when the search is restricted to replacing between two OCSes.

Consider the process of the BFS described in \ref{resttwo}. First, assume $m$ is the number of input/output switches in the traditional model. In each layer, there are $c$ enumerable connections, where $c$ is the link capacity. Ideally, these $c$ connections lead to different links, and the BFS always expands to unvisited links. In this case, the length of the replacement chain is upper bounded by $\log_cm$. When $m=256$ and $c=4$, $\log_cm=4$. This upper bound is based on ideal assumptions and is usually not achievable, but it still reflects that the length of the shortest replacement chain is very limited.

Furthermore, it also implies that the length of the shortest replacement chain decreases as $c$ increases. We conducted an experiment on the maximum length of the shortest replacement chain in Facebook's traffic data, under the network parameters of 100\% network load, 256 OCSes, $m=155$, and $c\in\{1,2,4,8\}$. The results presented in Table~\ref{tab3} are consistent with our speculation.

\begin{table}[!t]
\caption{Maximum lengths of the shortest replacement chain under different link capacities in an experiment.}\label{tab3}
\begin{center}
\begin{tabular}{c|cccc}
\toprule
Link capacity ($c$)&$1$&$2$&$4$&$8$\\
\midrule
\emph{Replacement chain length}&$138$&$10$&$7$&$5$\\
\bottomrule
\end{tabular}
\end{center}
\end{table}

%-----------------------------------
\subsection{Algorithm Adaptiveness and Flexibility}\label{adap}
%-----------------------------------

The algorithm itself is very flexible and can support functional requirements in actual scenarios that other algorithms cannot support. For example, to prevent a single search from taking longer than the average due to excessive search breadth and depth in extreme cases, we can gradually limit the search breadth as the depth limit grows. Specifically, given the maximum number of computations \(\mathrm{maxComp}\) to achieve and the current depth limit \(d\), we limit the breadth at each node of the search tree to:
\begin{equation}
\mathrm{maxBreadth} = \mathrm{maxComp}^{\frac1d}
\end{equation}
so that each depth-deepening iteration has approximately the same amount of computation equal to \(\mathrm{maxComp}\). When the number of tried selectable switches reaches the breadth limit, the algorithm stops expanding on more selectable switches. This operation actually makes the algorithm adaptively trade off running time and the number of rewirings, adding a layer of protection for the stability of algorithm performance. This technique is validated to be effective in at least all the artificially constructed extreme cases we tested.

Our algorithm also works on a first-come, first-served basis, which means that it naturally supports some priority-related requirements. For example, in irregularly-structured networks, it is difficult to determine whether a demand is feasible, so connections that are more important can be scheduled first, so that connections with higher priorities are established with greater probability.

%-----------------------------------
\subsection{Refining Across Multiple Schedulings}\label{refinedver}
%-----------------------------------

The original algorithm always finds the shortest possible replacement chain for scheduling a single connection. When scheduling multiple connections, we can further improve the algorithm. We consider that there may be some operations in the later scheduling of connections that cancel each other out with some operations in the previous scheduling of connections. Therefore, we can search for multiple replacement chains and select the one with the smallest total number of rewirings when combined with all previous replacement chains. This idea is somewhat similar to the Successive Shortest Path (SSP) algorithm in the MCF problem, if you consider the replacement chain as a special kind of augmenting path.

In our experiments, however, we found that the refined version has some relatively small improvements in the number of rewirings, but has a larger negative impact on the constant factor of the time complexity (please refer to Section~\ref{coneval} for some numerical statistics). Since this is essentially a trade-off between running time and the number of rewirings, we leave the choice between these two variants (more specifically, the choice of the hyperparameter of the number of replacement chains to be searched) to the actual network designer.

% \empty
% {\noindent \large \bf \color{red} !THE FOLLOWING NOT YET REVISED}

%-------------------------------------------------------------------------------
%-------------------------------------------------------------------------------
\section[(Concentrated) ToE Evaluation]{(Concentrated)\footnote{``Concentrated'' ToE refers to the ToE commonly adopted in current network architectures, where OCSes are reconfigured periodically, as the opposite of ``real-time'' ToE.} ToE Evaluation}\label{coneval}
%-------------------------------------------------------------------------------

%-----------------------------------
\subsection{Network Structure Adaptation}\label{modeladap}
%-----------------------------------

There exist multiple ToE algorithms and general scheduling algorithms\footnote{``General scheduling algorithms'' refer to those algorithms that are essentially the same as ToE algorithms, but without the objective of minimizing the number of rewirings (which is referred to as ``static centralized scheduling algorithms'' in \cite{cao2021research}). These algorithms were typically studied in the context of traditional Clos networks.} \cite{cao2021research,hwang2003control,hwang2002modification,lee1996new,wang2018parallel} for the traditional model. However, most of them are based on the bipartite graph structure and therefore cannot be applied directly to the bidirectional model, thus a mapping method to reduce to the traditional model is required. It has been mentioned in Section~\ref{curdesign} that such a mapping is imperfect and introduces a mapping loss. Similar ideas has also been discussed in \cite{han2024lumoscore}.

We assume that the bidirectional model we use has even link capacities. To convert the bidirectional network structure and the port matching schemes to traditional ones, we use the bipartite coloring method mentioned in Section~\ref{existrc}. Additionally, we need an algorithm to convert the logical topologies in the bidirectional model to the traditional model.

Let \(D\) be the original logical topology in the bidirectional model and \(D'\in\mathbb{Z}^{m\times m}_{\ge0}\) be the converted logical topology. Let \(C\) be the link capacity matrix in the bidirectional model. According to the previous analysis, we know that \(D\) is feasible if and only if it satisfies the following condition:
\begin{equation}
\forall j,\sum_kD_{j,k}\le\sum_iC_{i,j}
\end{equation}
We also know that \(D'\) is legal and feasible if and only if it satisfies the following conditions:
\begin{subequations}
\begin{equation}
\forall j,k,D'_{j,k}+D'_{k,j}=D_{j,k}    
\end{equation}
\begin{align}
\forall j,&\sum_kD'_{j,k}\le\frac12\sum_iC_{i,j}\\
\forall k,&\sum_jD'_{j,k}\le\frac12\sum_iC_{i,k}
\end{align}
\end{subequations}
We can tighten the constraints so that \(D'\) satisfies the following conditions:
\begin{subequations}
\begin{align}
\forall j,&\left\lfloor\frac12\sum_kD_{j,k}\right\rfloor\le\sum_kD'_{j,k}\le\left\lceil\frac12\sum_kD_{j,k}\right\rceil\\
\forall k,&\left\lfloor\frac12\sum_jD_{j,k}\right\rfloor\le\sum_jD'_{j,k}\le\left\lceil\frac12\sum_jD_{j,k}\right\rceil
\end{align}
\end{subequations}
We now show an algorithm that consistently achieves the above constraints.

Let \(G_{j,k}\) be the remainder of \(D_{j,k}\) divided by \(2\). As a result, \(G\) can be seen as the adjacent matrix of a simple undirected graph consisting of \(m\) vertices.

First, we repeat the following process: select a vertex with an odd degree in the graph, and arbitrarily select a trail (i.e. a path with possibly repeated nodes but not repeated edges) starting from this vertex until we reach some other vertex with an odd degree for the first time. Before reaching that vertex, there always exist unvisited edges to go through, since the vertex degrees are even.

Then, orient the edges on the trail in the direction of the traversal and delete them from the undirected graph. It can be seen that in a trail, there is exactly one edge connected to the starting vertex and one edge connected to the ending vertex that are oriented, and for other vertices, there are exactly two edges connected to them that are oriented in opposite directions. Since there must be an even number of vertices with odd degrees in an undirected graph, the above process will turn all vertices with odd degrees into vertices with even degrees.

Second, we repeat the same process, but starting from and ending at the same even-degree vertex. The process is trivial.

After these processes, all edges connected to each vertex in the graph are oriented in either direction, and the difference between the number of edges in the two directions does not exceed \(1\). Let \(G'\) be the adjacency matrix of the directed graph after all edges are directed, then let
\begin{equation}
D'_{j,k}=\left\lfloor\frac 12D_{j,k}\right\rfloor+G'_{j,k}
\end{equation}
It can be found that \(D'\) satisfies the constraints imposed above. For a detailed proof, see Appendix~\ref{nsaproof}.

The above algorithm can be easily implemented in \(O\left(m^2\right)\) time complexity.

%-----------------------------------
\subsection{Evaluated Algorithms}\label{evalalg}
%-----------------------------------

\noindent \textbf{Brute Force ILP.} As is done by Zhao et al. \cite{zhao2019minimal}, the problem can be modeled as an ILP problem and then solved using a solver such as Gurobi \cite{gurobi}. However, given the scale of the problem and the efficiency of ILP solvers, this approach is impractical in medium- to large-scale clusters and often takes days without finding a solution. Due to this inability, this algorithm is {\color{red}\textbf{NOT}} included in the results.

\noindent \textbf{Greedy MCF.} This is an algorithm introduced in \cite{zhao2019minimal}. It essentially divides the logical topology equally among all OCSes and then uses MCF to distribute the remainders. Since this algorithm is based on MCF, it only supports the traditional model under 100\% network load (for the definition of network load, see \ref{evalsetup}).

\noindent \textbf{Bipartition MCF.} The bipartition method was introduced in Google's patent \cite{Wang_Zhang_Li_2022}. This method requires solving a subproblem in each bipartition step. In \cite{zhang2023reducing}, the subproblem is solved using MCF. This algorithm also supports only the traditional model under 100\% network load. We extended this approach to partially loaded networks using a technique we call minimum-cost bounded flow (MCBF).

\noindent \textbf{Bipartition MCBF.} This is the algorithm extended from Bipartition MCF with the technique of MCBF to encompass partially loaded networks.

\noindent \textbf{FastReChain.} This is our algorithm described in Section~\ref{thealg} with bitset optimization.

\noindent \textbf{FastReChainRefined.} This is the refined version of FastReChain described in Section~\ref{refinedver}. This algorithm improves in terms of the number of rewirings by searching for multiple replacement chains and selecting the best one. The more replacement chains are searched each time, the fewer rewirings, but the search takes more time. In our evaluation, we set the number of replacement chains that the algorithm searches each time to 8.

%-----------------------------------
\subsection{Evaluation Setup}\label{evalsetup}
%-----------------------------------

We will evaluate the algorithms in both the bidirectional and the traditional model. Therefore, we first need to redefine some notations in the context of the traditional model. A traditional model contains \(m\) input switches, \(m\) output switches, and \(n\) OCSes. \(D_{j,k}\) represents the number of connections required between the \(j\)-th input switch and the \(k\)-th output switch. \(X_{i,j,k}\) represents the number of connections between the \(j\)-th input switch and the \(k\)-th output switch that pass through the \(i\)-th OCS in scheme \(X\). \(C_{i,j}\) represents the link capacity between the \(i\)-th OCS and the \(j\)-th input switch as well as between the \(i\)-th OCS and the \(j\)-th output switch. In the following text, unless otherwise specified, all the notations apply to both networks.

\noindent \textbf{Simulation scenario.} We assume continuous monitoring of the traffic between ToR switches (or between input switches and output switches). Our objective is to generate the logical topology at a fixed time interval and reconfigure the OCSes accordingly to enable the network configuration to adapt to changes in the traffic pattern.

\noindent \textbf{Task modes.} In each test case, there will be a series of demand phases arranged on a timeline, where in each phase a new logical topology arrives.

We will evaluate the algorithms in two task modes. The first is the continuous task mode, in which the original matching scheme is the scheme generated by the algorithm itself last demand phase. For the initial demand phase, the original matching scheme is set to empty (i.e. no connections are established). The second is the discontinuous task mode, in which the original matching scheme is randomly constructed based on the last logical topology. This is to simulate the situation where some unexpected errors happen and another algorithm has to take over the task. Indeed, we found that some of the algorithms have the ability to exploit their own scheme structures, so we need to test their performance in different scenarios.

\noindent \textbf{Result measurement.} We measure the quality of our schemes in terms of \emph{rewiring ratio} (RR). There are $T$ demand phases $D^{(0)},D^{(1)},\ldots,D^{(T-1)}$, and $T-1$ times of reconfiguration in between these demand phases, where the $t$-th reconfiguration consists of an original matching scheme $X^{(t)}$, and a new matching scheme $Y^{(t)}$ produced by the algorithm. Note that $X^{(t)}$ conforms to $D^{(t)}$ and $Y^{(t)}$ conforms to $D^{(t+1)}$. We define the number of connections in $D^{(t)}$ as
\begin{equation}
\mathrm{NumConns}(t)=\sum_{j=0}^{m-1}\sum_{k=0}^{m-1}D^{(t)}_{j,k}
\end{equation}
We define the number of rewirings between $X^{(t)}$ and $Y^{(t)}$ as
\begin{equation}
\mathrm{NumRewirings}(t)=\sum_{i=0}^{n-1}\sum_{j=0}^{m-1}\sum_{k=0}^{m-1}\left\vert X^{(t)}_{i,j,k}-Y^{(t)}_{i,j,k}\right\vert
\end{equation}
where $\vert x\vert$ represents the absolute value of $x$. Then, the rewiring ratio of the $t$-th reconfiguration phase is defined as
\begin{equation}
\mathrm{RR}(t)=\frac{\mathrm{NumRewirings}(t)}{\mathrm{NumConns}(t)+\mathrm{NumConns}(t+1)}
\end{equation}
The smaller the rewiring ratio, the better.

\noindent \textbf{Data source.} We use the open traffic traces from Facebook datacenters in 2024 \cite{avin2020complexity}. Given the traffic matrix \(\mathrm{Tr}\) aggregated in a time period, where \(\mathrm{Tr}_{j,k}\) represents the volume of traffic from the \(j\)-th rack to the \(k\)-th rack, the logical topology is converted from \(\mathrm{Tr}\) according to the following prioritization strategy:

Specifically, for the traditional model, the \(r\)-th connection between the \(j\)-th input switch and the \(k\)-th output switch has a weight of \(\left(\mathrm{Tr}_{j,k}+1\right)/r\). For the bidirectional model, the \(r\)-th connection between \(L_j\) and \(L_k\) has a weight of \(\left(\max\left\{\mathrm{Tr}_{j,k},\mathrm{Tr}_{k,j}\right\}+1\right)/r\). The connection with the highest weight which can be added to the logical topology matrix without making it infeasible is added, until no more connections satisfy the condition, or the desired network load is achieved.

We choose an aggregation interval length of 3600 seconds and generate a logical topology every 600 seconds. That is, the time intervals corresponding to the demand phases (in seconds) are \([0,3599],[600,4199],\ldots\). Please note that this interval is set only for the convenience of benchmarking, and the actual interval can be set smaller. Each traffic trace will provide 139 demand phases, including the initial demand phase.

\noindent \textbf{Parameters.} We use the trace of Facebook cluster A in 2024, which has \(m=155\). Define \(c\) as the link capacity between each pair of switches and \(l\) as the network load, which is defined as the sum of the used capacity of each link divided by the sum of the capacity of each link, i.e., \(l\) is defined as
\begin{equation}
l=\frac{\sum_{j,k}D_{j,k}}{\sum_{i,j}C_{i,j}}
\end{equation}
We tested combinations of \(n\in\{128,256,384\}\), \(c\in\{2,4,8\}\) for the traditional model or \(c\in\{4,8,16\}\) for the bidirectional model, and \(l\in\{0.2,0.4,0.6,0.8,1.0\}\).

%-----------------------------------
\subsection{Evaluation Results}\label{evalres}
%-----------------------------------

\begin{figure}[!t]
    \centering

    \includegraphics[width=0.9\linewidth]{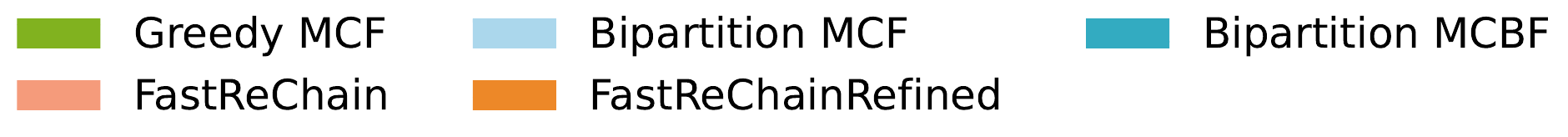}

    \begin{tabular}{@{}c@{}c@{}c@{}}
    
    \subfloat[\(n=128,c=2\)]{
    \includegraphics[width=0.3\linewidth]{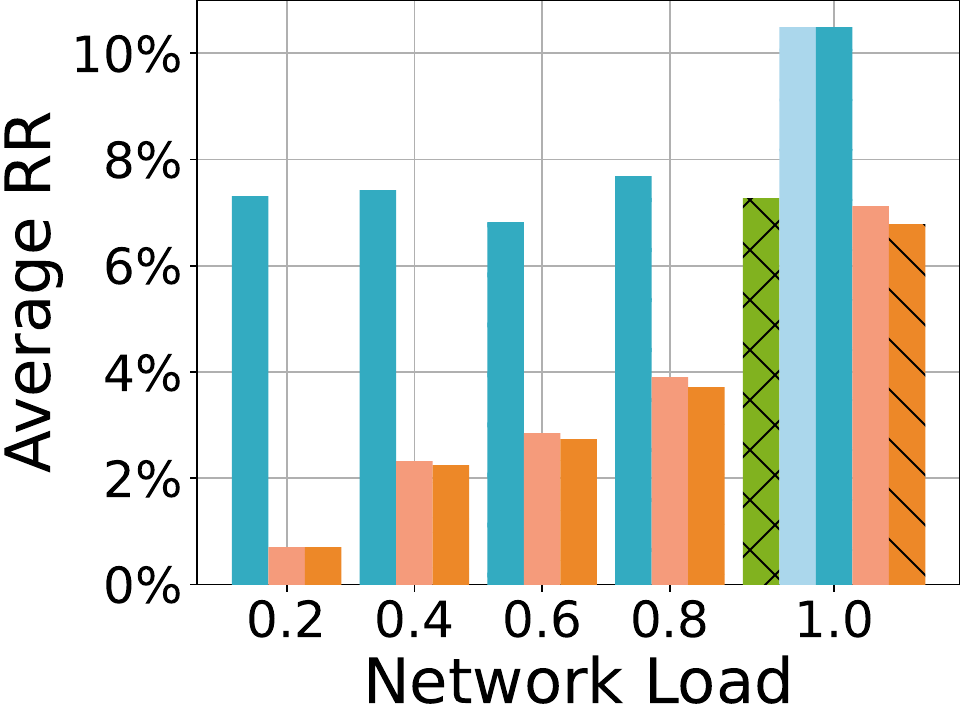}
    }
    &
    \subfloat[\(n=128,c=4\)]{
    \includegraphics[width=0.3\linewidth]{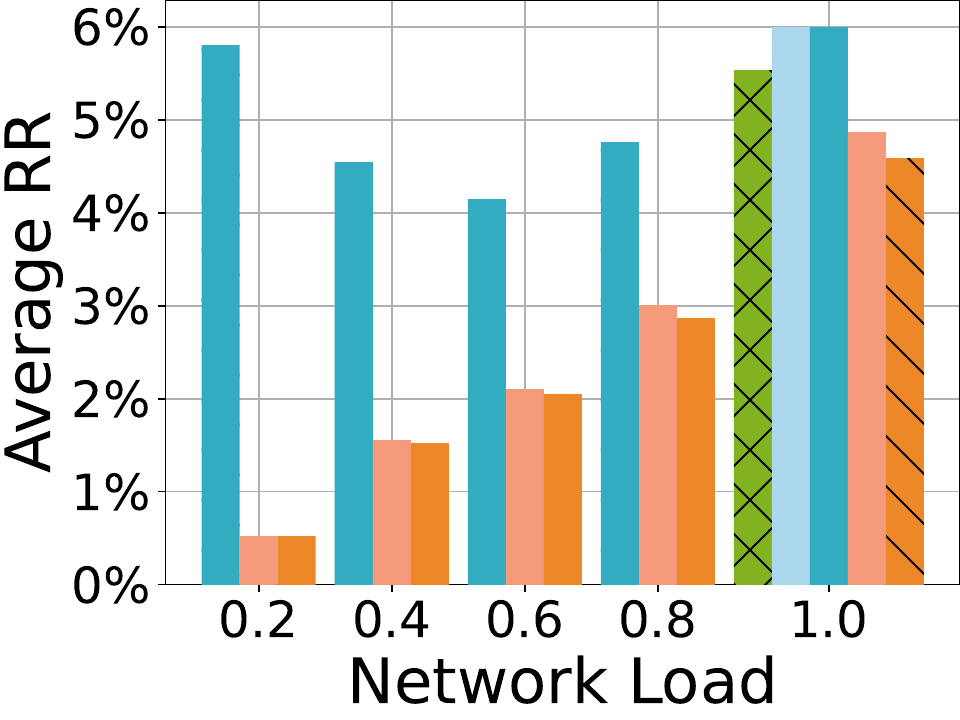}
    }
    &
    \subfloat[\(n=128,c=8\)]{
    \includegraphics[width=0.3\linewidth]{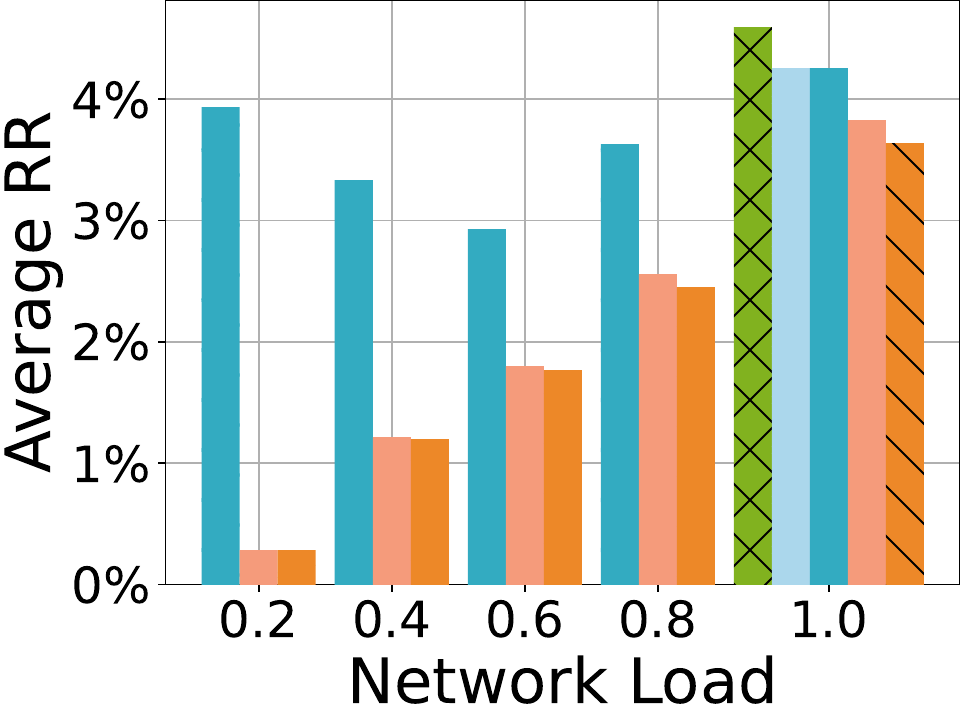}
    }

    \\
    
    \subfloat[\(n=256,c=2\)]{
    \includegraphics[width=0.3\linewidth]{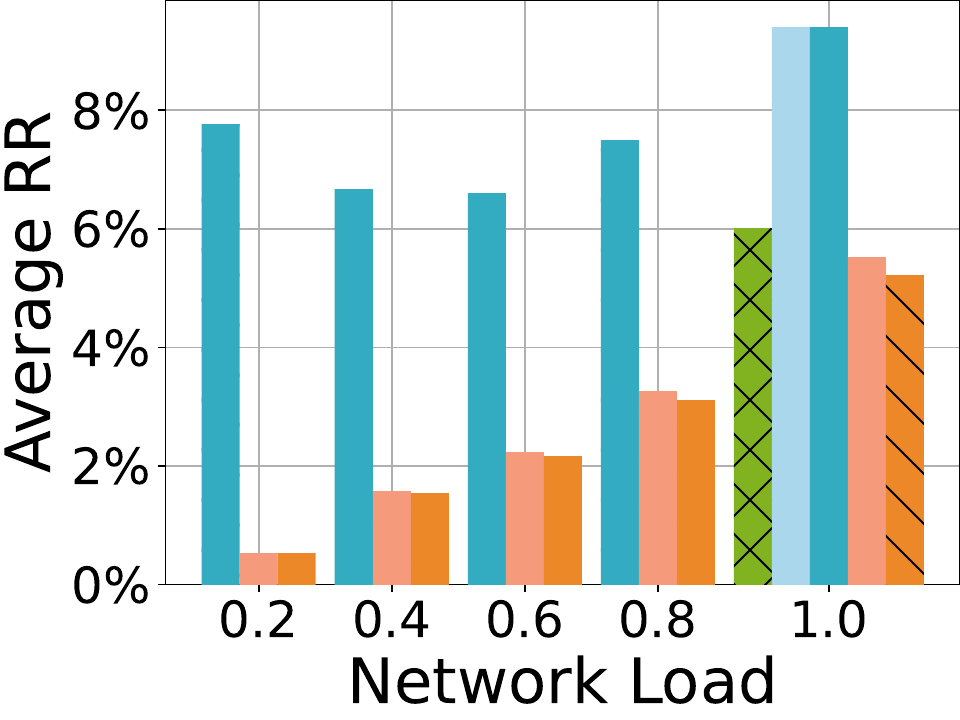}
    }
    &
    \subfloat[\(n=256,c=4\)]{
    \includegraphics[width=0.3\linewidth]{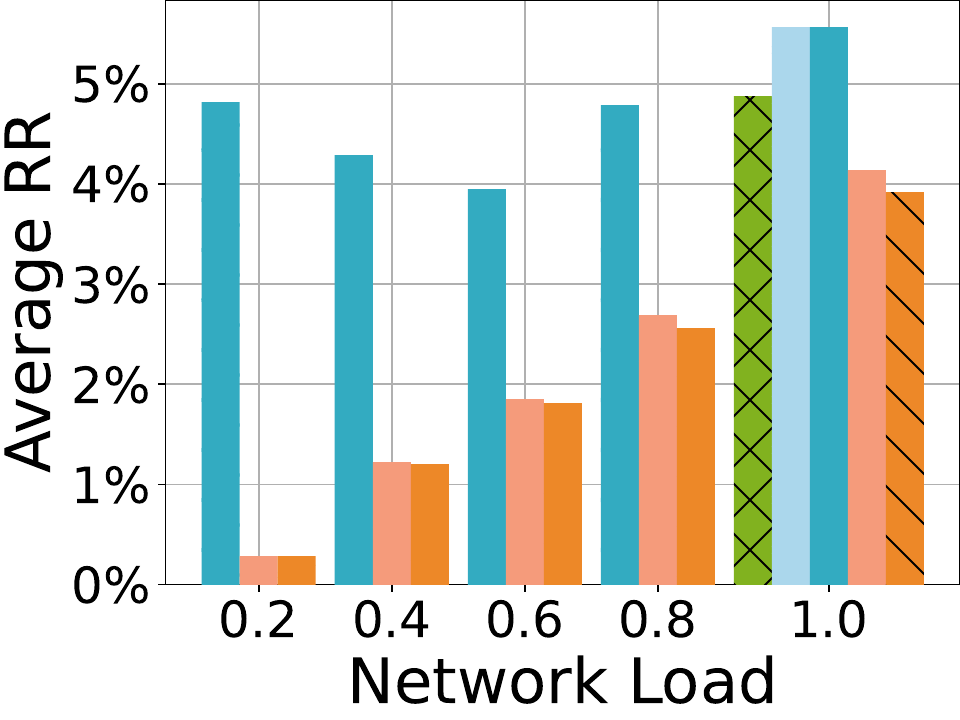}
    }
    &
    \subfloat[\(n=256,c=8\)]{
    \includegraphics[width=0.3\linewidth]{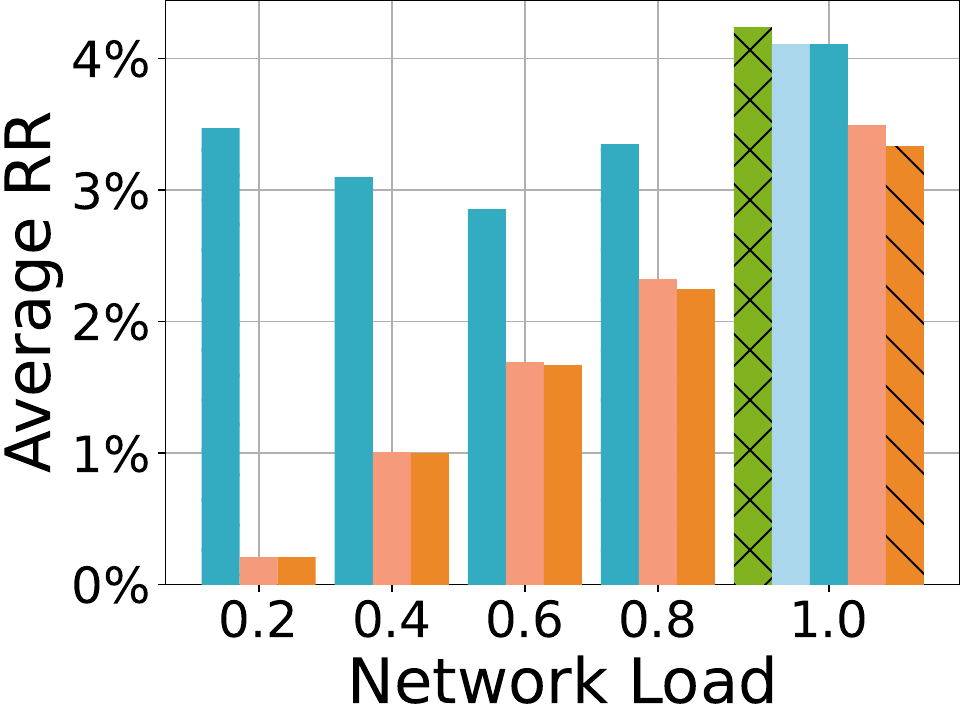}
    }

    \\
    
    \subfloat[\(n=384,c=2\)]{
    \includegraphics[width=0.3\linewidth]{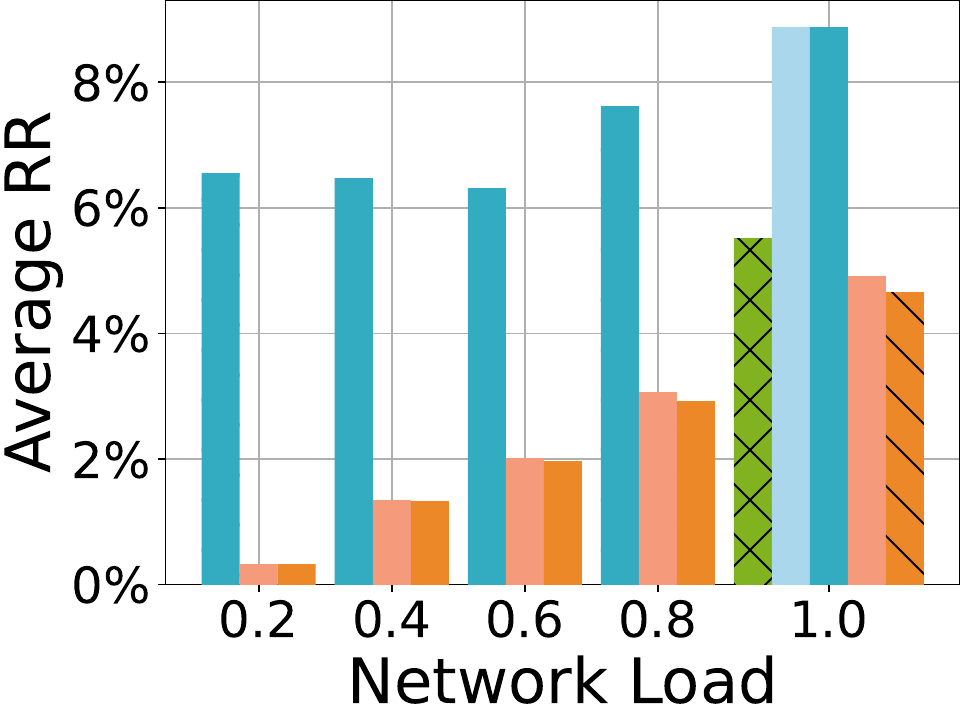}
    }
    &
    \subfloat[\(n=384,c=4\)]{
    \includegraphics[width=0.3\linewidth]{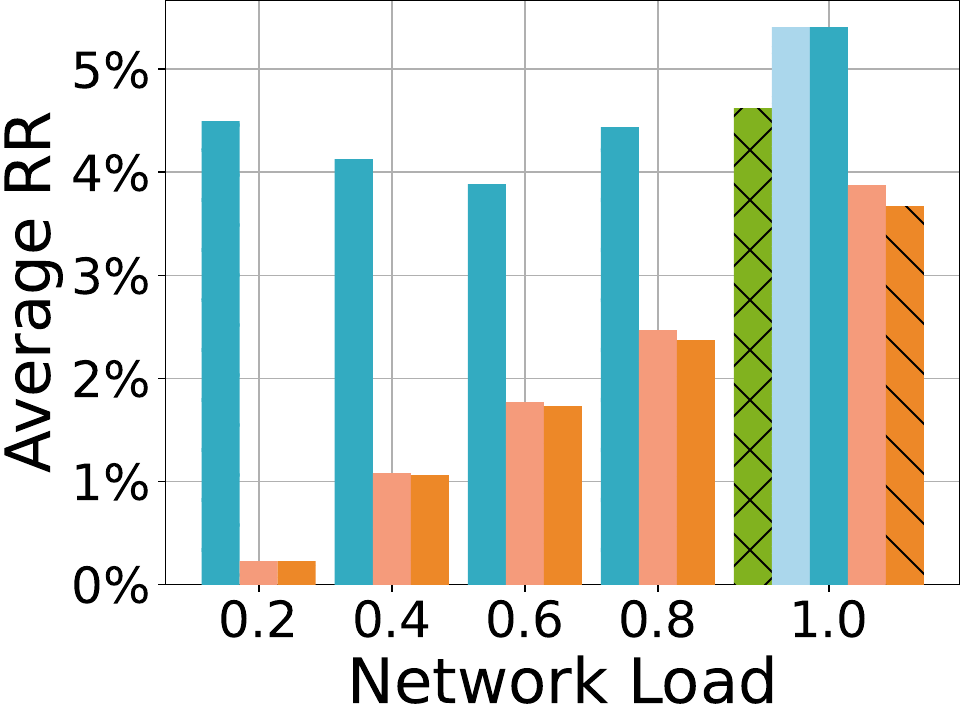}
    }
    &
    \subfloat[\(n=384,c=8\)]{
    \includegraphics[width=0.3\linewidth]{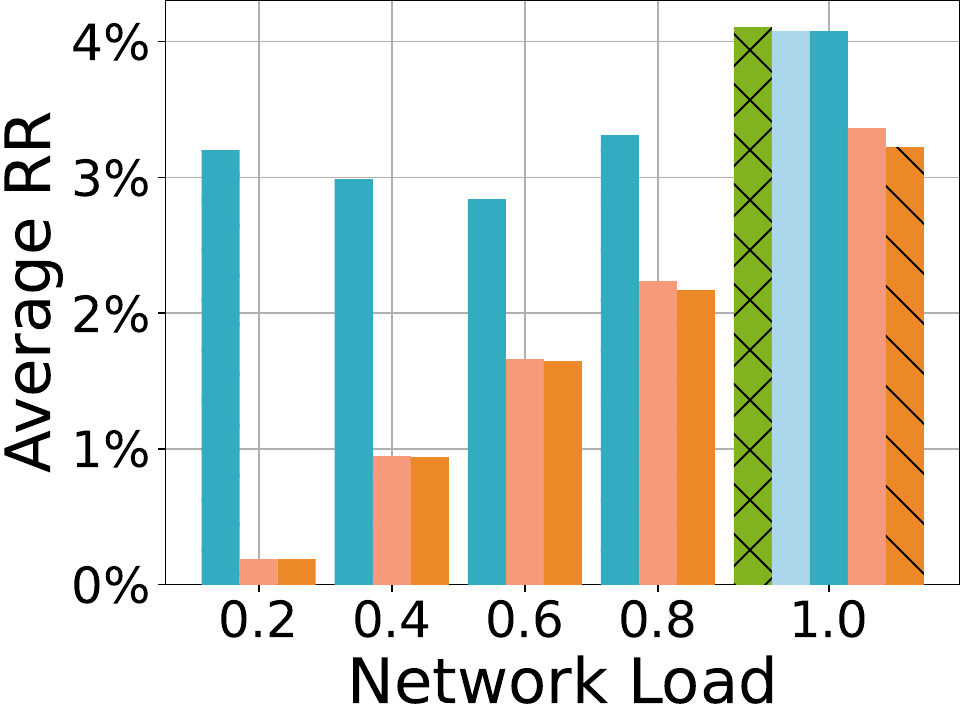}
    }

    \end{tabular}
    
    \caption{Rewiring ratio in the traditional model, continuous task mode.}\label{figACRR}
\end{figure}

\begin{figure}[!t]
    \centering
    
    \includegraphics[width=0.9\linewidth]{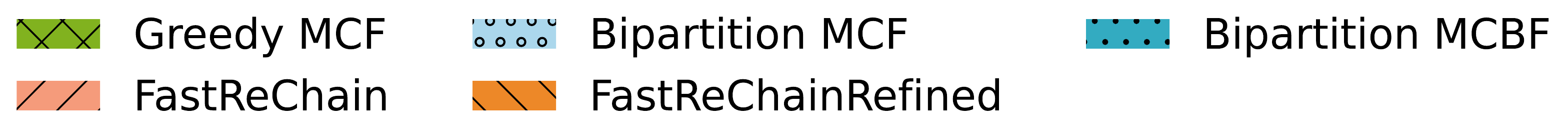}

    \begin{tabular}{@{}c@{}c@{}c@{}}
    
    \subfloat[\(n=128,c=2\)]{
    \includegraphics[width=0.3\linewidth]{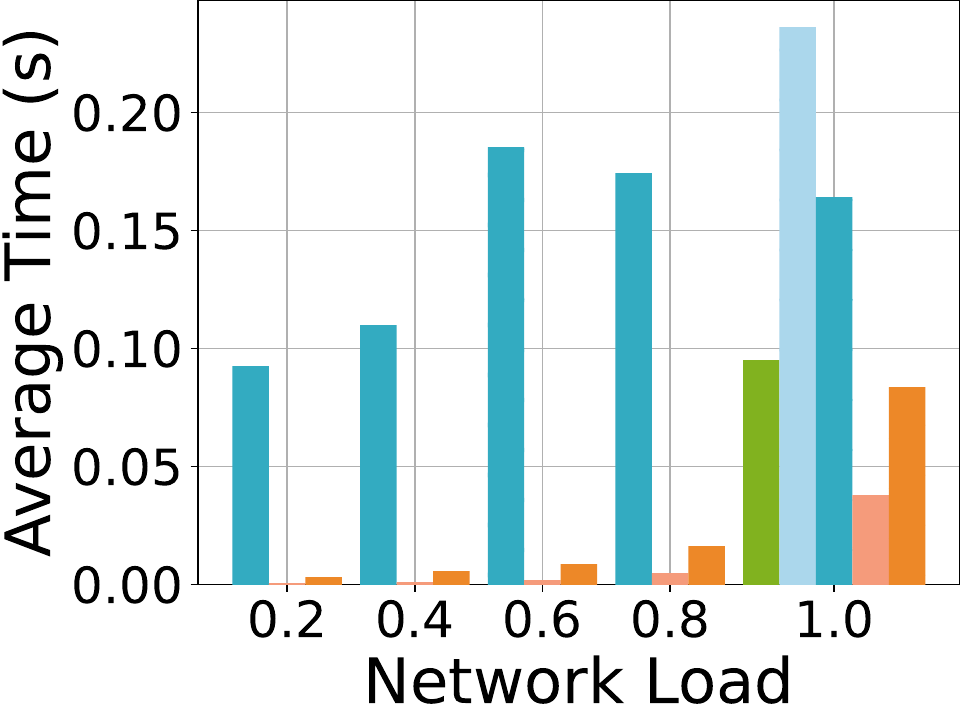}
    }
    &
    \subfloat[\(n=128,c=4\)]{
    \includegraphics[width=0.3\linewidth]{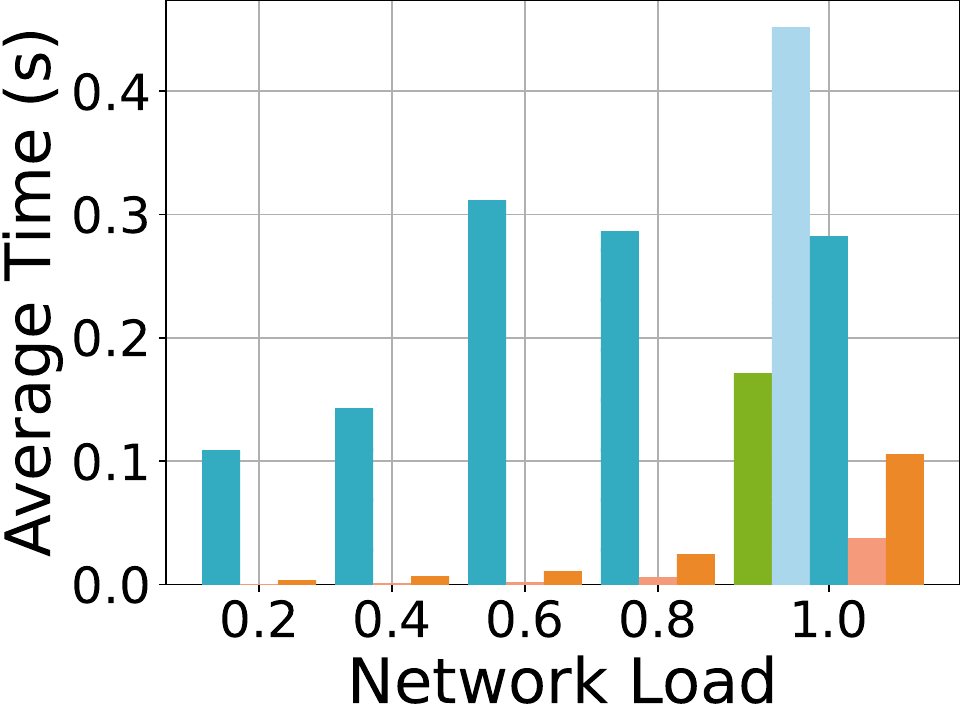}
    }
    &
    \subfloat[\(n=128,c=8\)]{
    \includegraphics[width=0.3\linewidth]{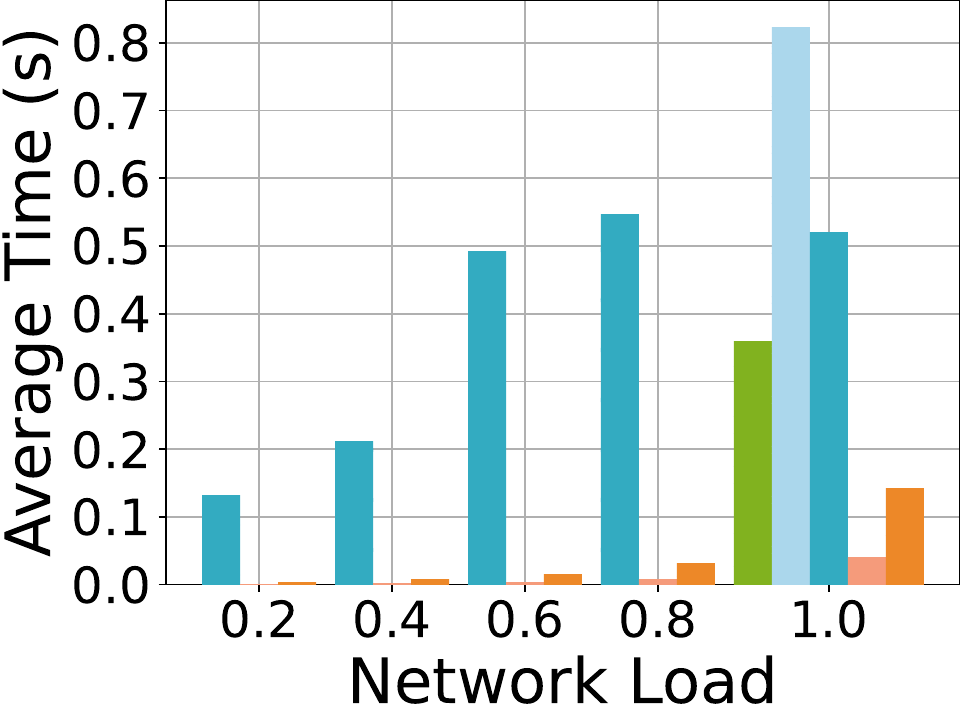}
    }

    \\
    
    \subfloat[\(n=256,c=2\)]{
    \includegraphics[width=0.3\linewidth]{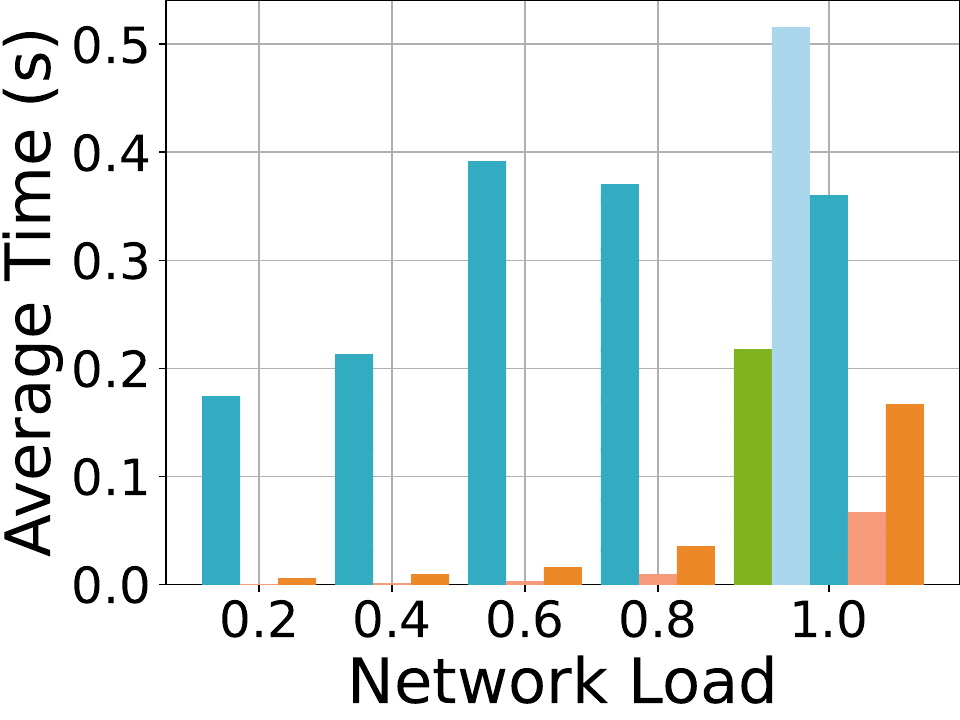}
    }
    &
    \subfloat[\(n=256,c=4\)]{
    \includegraphics[width=0.3\linewidth]{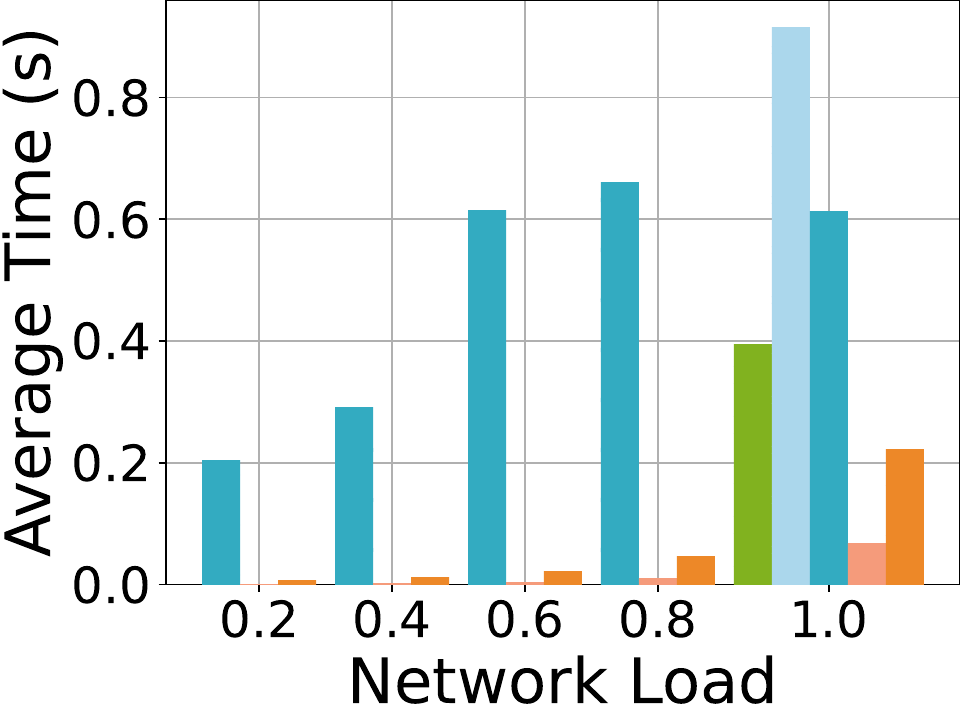}
    }
    &
    \subfloat[\(n=256,c=8\)]{
    \includegraphics[width=0.3\linewidth]{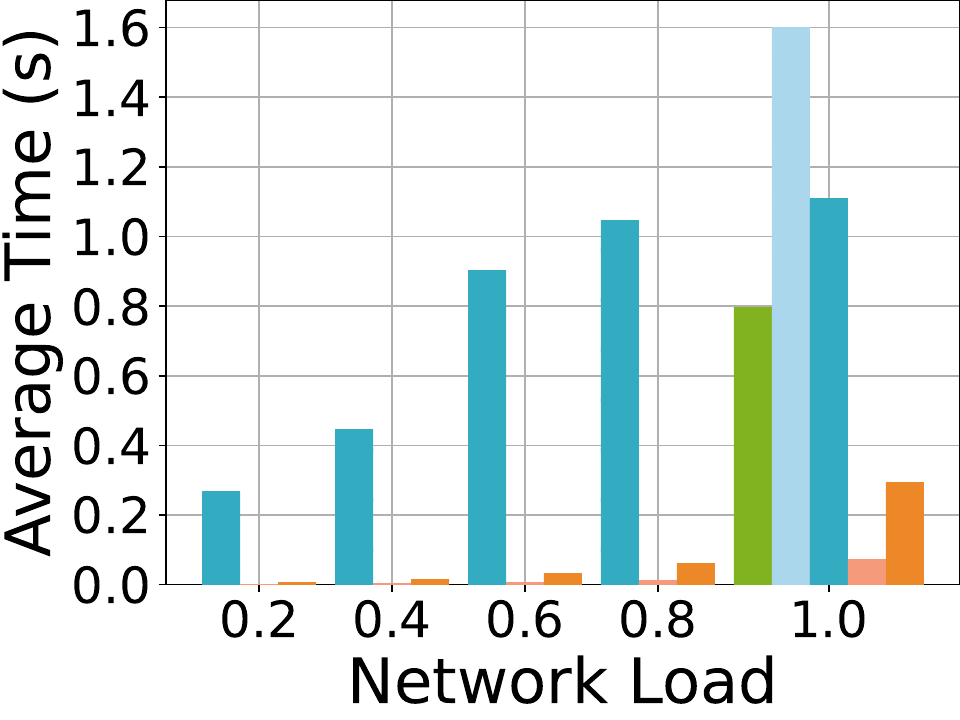}
    }

    \\
    
    \subfloat[\(n=384,c=2\)]{
    \includegraphics[width=0.3\linewidth]{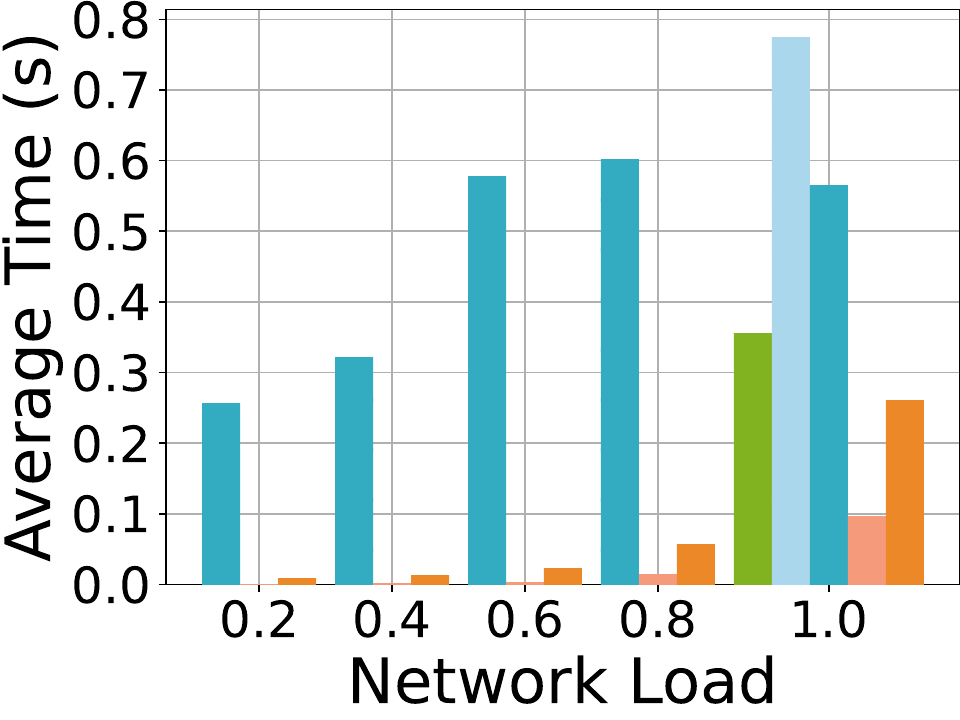}
    }
    &
    \subfloat[\(n=384,c=4\)]{
    \includegraphics[width=0.3\linewidth]{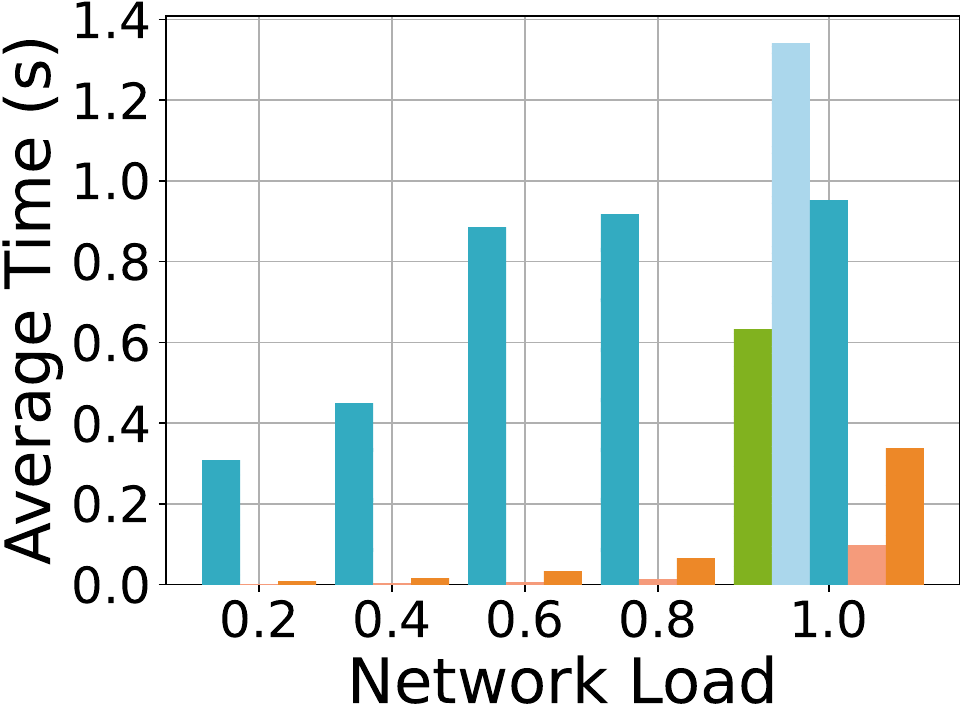}
    }
    &
    \subfloat[\(n=384,c=8\)]{
    \includegraphics[width=0.3\linewidth]{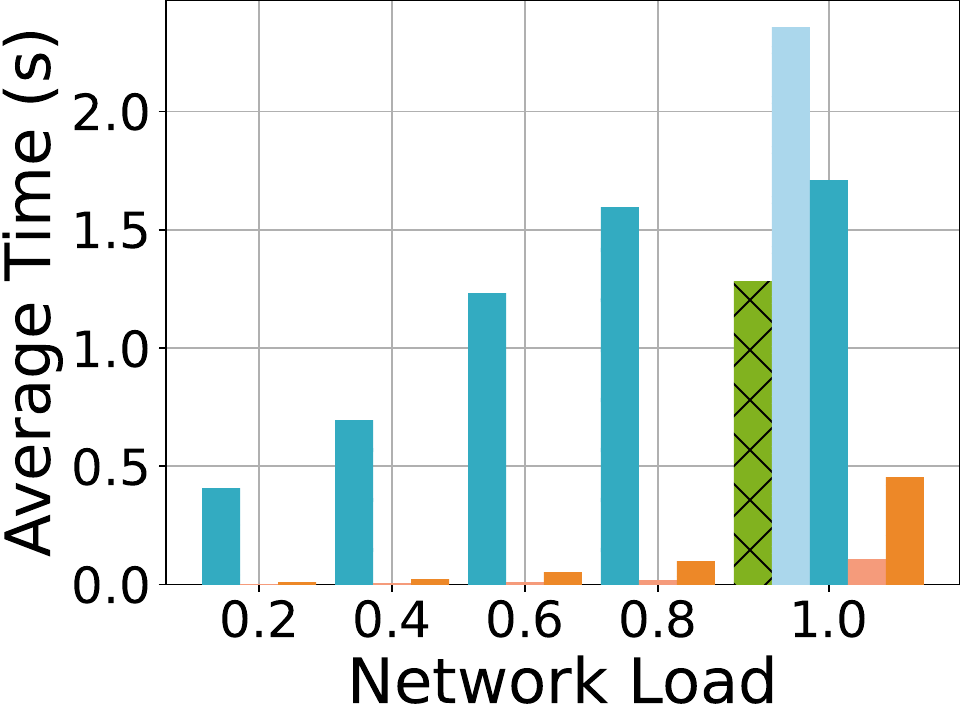}
    }
    
    \end{tabular}
    
    \caption{Running time in the traditional model, continuous task mode.}\label{figACTime}
\end{figure}

Figs.~\ref{figACRR}~and~\ref{figACTime} show the rewiring ratios and the running times in the traditional model in continuous task mode, respectively. In general, FastReChain outperforms other algorithms in all metrics across all scenarios. In terms of running time, FastReChain outperforms the others overwhelmingly. Even in the worst case of $n=128$, $c=2$ and $l=1$, FastReChain only takes less than half the time of Greedy MCF, and in the vast majority of cases with $l=1$, FastReChain takes about a quarter or even less of the time of Greedy MCF, not to mention other algorithms. When the network load is low, FastReChain overwhelmingly outperforms Bipartition MCBF, taking just a few milliseconds, while Bipartition MCBF requires hundreds of milliseconds. In terms of rewiring ratio, when $l=1$, FastReChain still has a solid improvement over other algorithms. When the network load is low, FastReChain's improvement is more significant, with rewiring rates that are about half or less of that of Bipartition MCBF.

At the same time, we find that in terms of rewiring ratio, the larger $c$ is and the smaller $n$ is, the gap between FastReChain and the bipartition methods narrows, but the performance of Greedy MCF becomes worse. We also find that in terms of running time, the larger $c$ and the larger $n$, the gap between FastReChain and other methods widens.

\begin{figure}[!t]
    \centering

    \includegraphics[width=0.9\linewidth]{expfigs/legend.pdf}

    \begin{tabular}{@{}c@{}c@{}c@{}}
    
    \subfloat[\(n=256,c=2\)]{
    \includegraphics[width=0.3\linewidth]{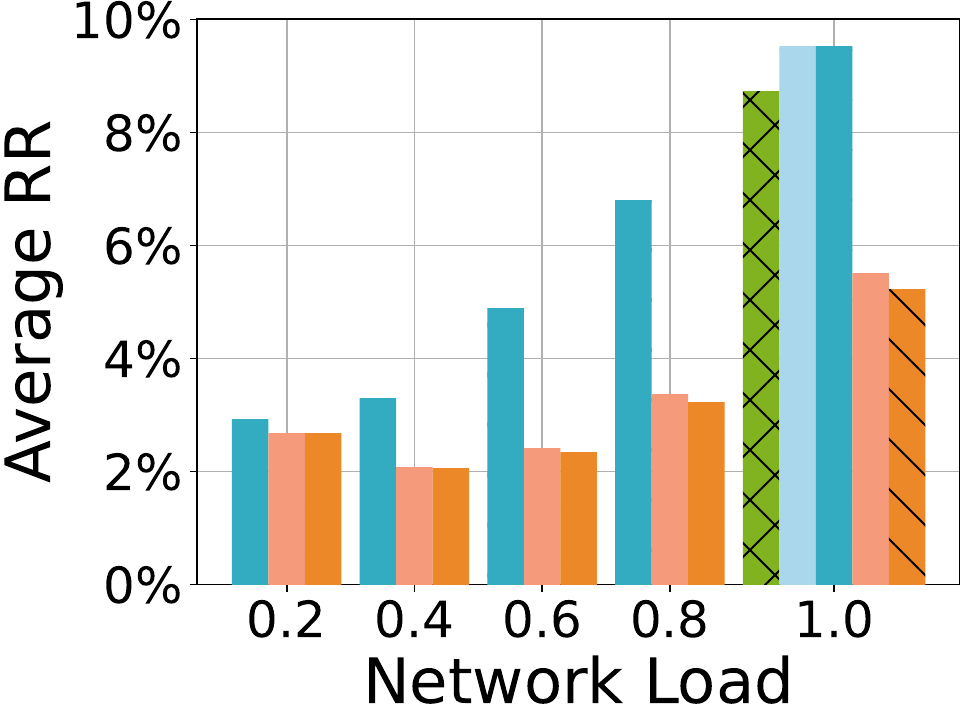}
    }
    &
    \subfloat[\(n=256,c=4\)]{
    \includegraphics[width=0.3\linewidth]{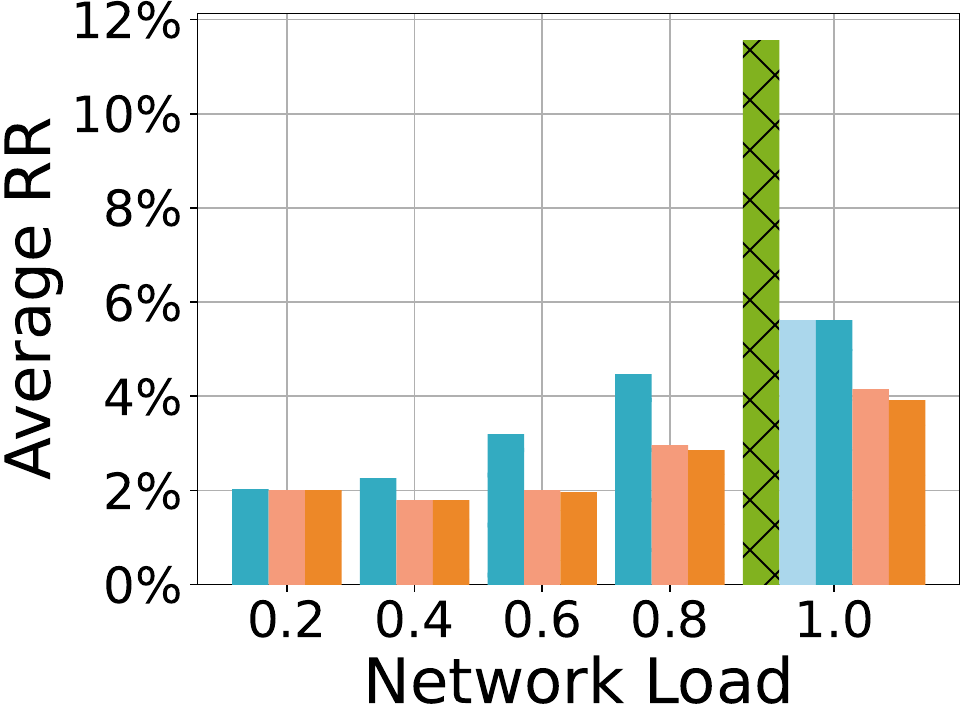}
    }
    &
    \subfloat[\(n=256,c=8\)]{
    \includegraphics[width=0.3\linewidth]{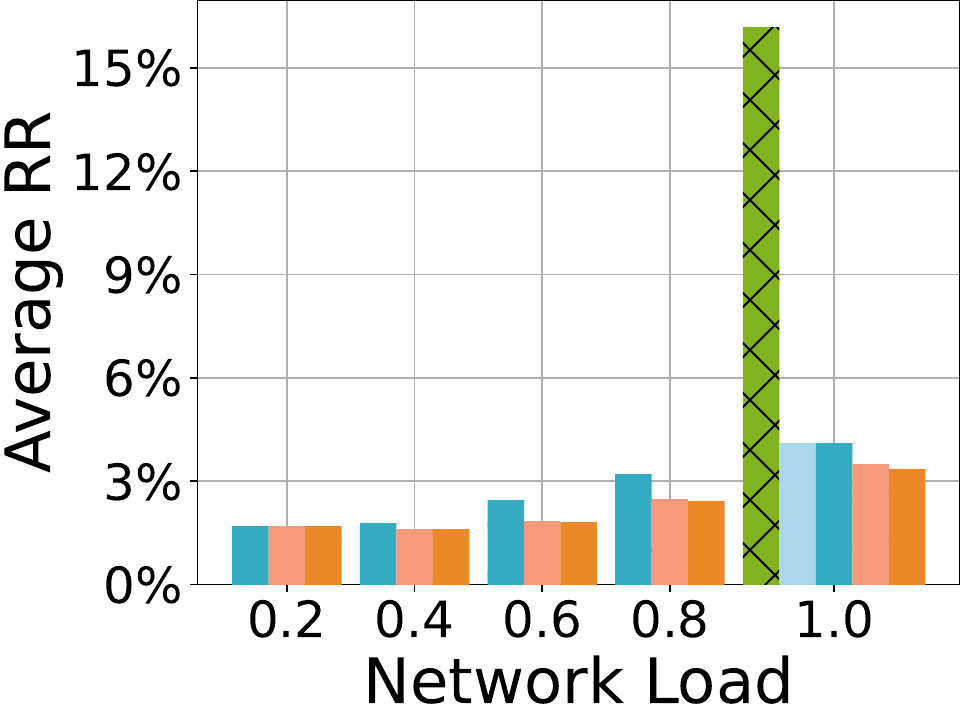}
    }

    \end{tabular}
    
    \caption{Rewiring ratio in the traditional model, discontinuous task mode.}\label{figADRR}
\end{figure}

For Fig.~\ref{figADRR} and Fig.~\ref{figBCRR}, we only show the results for the case of $n=256$ because the trends of the results in other cases are similar.

Fig.~\ref{figADRR} shows the rewiring ratios in the traditional model in discontinuous task mode. We find that in this case, the gap between FastReChain and Bipartition MCBF in terms of rewiring ratio becomes smaller when the network load is low, but the performance of Greedy MCF deteriorates significantly as $c$ increases.

\begin{figure}[!t]
    \centering

    \includegraphics[width=0.9\linewidth]{expfigs/legend.pdf}

    \begin{tabular}{@{}c@{}c@{}c@{}}
    
    \subfloat[\(n=256,c=4\)]{
    \includegraphics[width=0.3\linewidth]{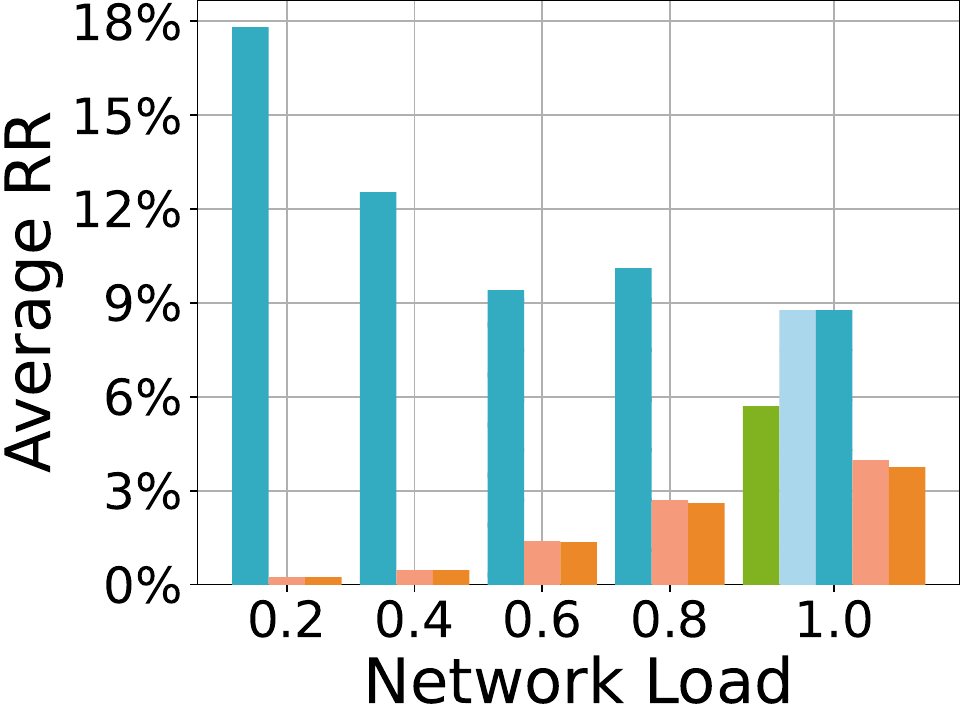}
    }
    &
    \subfloat[\(n=256,c=8\)]{
    \includegraphics[width=0.3\linewidth]{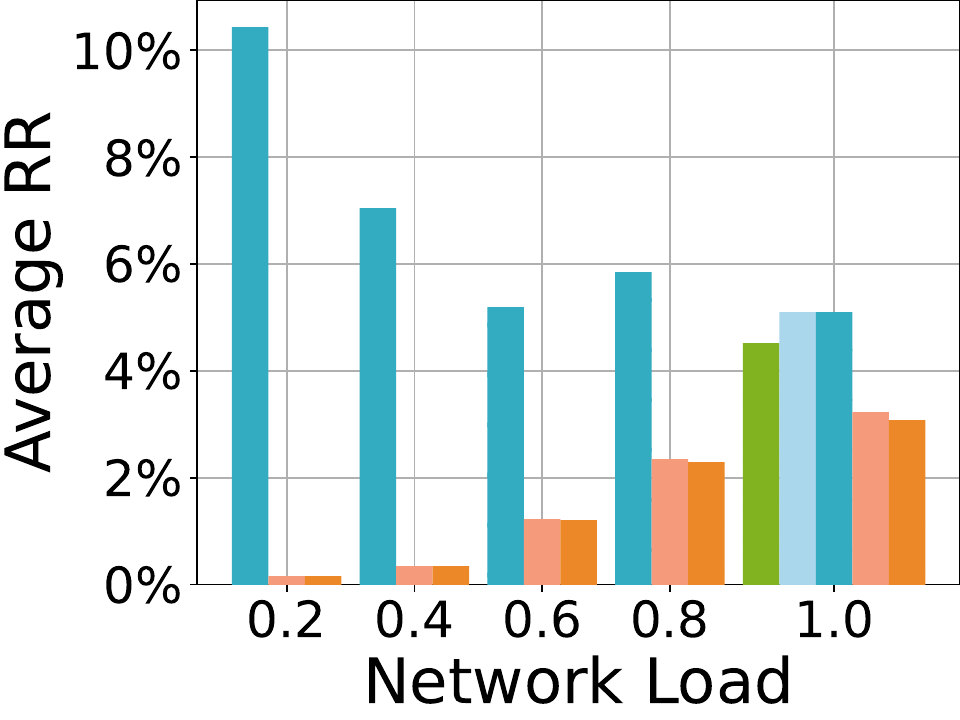}
    }
    &
    \subfloat[\(n=256,c=16\)]{
    \includegraphics[width=0.3\linewidth]{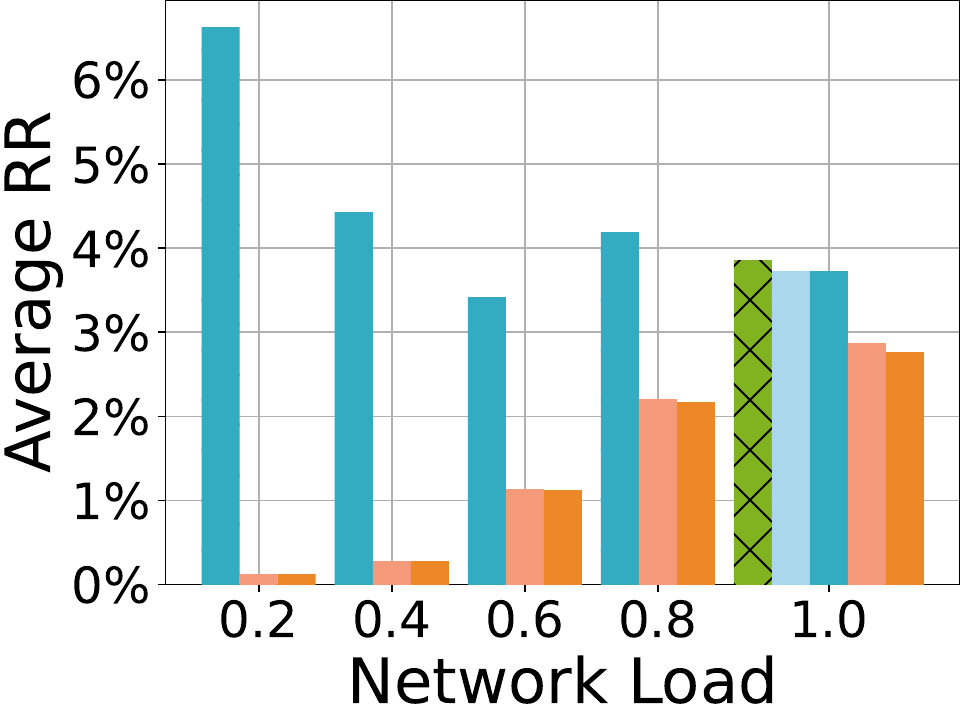}
    }
    
    \end{tabular}
    
    \caption{Rewiring ratio in the bidirectional model, continuous task mode.}\label{figBCRR}
\end{figure}

Fig.~\ref{figBCRR} shows the rewiring ratios in the bidirectional model in continuous task mode. Note that for all other algorithms, network structure adaption is applied. We can see that in bidirectional networks, the gap between FastReChain and other algorithms in terms of rewiring ratio widened considerably.
% \input{[OLD]/Static Minimal Rewiring Evaluation}
%-------------------------------------------------------------------------------

%-------------------------------------------------------------------------------
%-------------------------------------------------------------------------------
\section{Real-time ToE Evaluation}\label{dseval}
%-------------------------------------------------------------------------------

\emph{Real-time ToE} is defined as opposed to \emph{concentrated ToE}, where the task changes from adapting to a whole new logical topology to supporting atomic changes in the logical topology (i.e., adding/removing 1 required connection between 2 devices), and constructing modification schemes correspondingly. It should be noted that this is potentially an important problem for future fast-switching optical switches \cite{ballani2018bridging,zhao2023optical}, but for now, it does not have a practical application.

To date, our algorithm is the only one that supports real-time ToE (or even dynamic centralized scheduling \cite{cao2021research}), so there are no other baseline algorithms. We evaluated real-time ToE only in the bidirectional model.

%-----------------------------------
\subsection{Evaluation Setup}
%-----------------------------------

\noindent \textbf{Simulation scenario.} The simulation scenario and the data source are similar to those in Section~\ref{evalsetup}. The difference is that in real-time ToE, the logical topology needs to change dynamically according to changes in network traffic. We define the following two types of atomic logical topology changes:
\begin{itemize}
    \item $\textrm{Add}(j,k)$: Add a connection between $L_j$ and $L_k$ in the logical topology, i.e. $D_{j,k}\leftarrow D_{j,k}+1,D_{k,j}\leftarrow D_{k,j}+1$.
    \item $\textrm{Remove}(j,k)$: Remove a connection between $L_j$ and $L_k$ in the logical topology, i.e. $D_{j,k}\leftarrow D_{j,k}-1,D_{k,j}\leftarrow D_{k,j}-1$.
\end{itemize}
The algorithm needs to generate corresponding matching scheme modifications based on each atomic logical topology change.

\noindent \textbf{Data source.} We still use the traces from Facebook datacenters in 2024, the same as in Section~\ref{evalsetup}. We need an algorithm different from that in Section~\ref{evalsetup} that can generate atomic logical topology changes dynamically. We still have a traffic matrix $\mathrm{Tr}$ aggregated in a time period, which is dynamically updated every second. We assign a priority weight of \(\left(\max\left\{\mathrm{Tr}_{j,k},\mathrm{Tr}_{k,j}\right\}+1\right)/r\) to the \(r\)-th connection between \(L_j\) and \(L_k\) in the demand. Each time a new packet between \(L_j\) and \(L_k\) arrives and is added to \(\mathrm{Tr}\), we try to add as many connections between \(L_j\) and \(L_k\) to the logical topology as possible, while removing connections with lower weights that hinder the addition. We choose an aggregation interval length of 600 seconds, and the atomic logical topology changes are dynamically generated each second.

\noindent \textbf{Parameters.} The parameters are the same as in Section~\ref{evalsetup}, except that the case \(l=1.0\) is not tested because the algorithm we use to generate logical topologies can only achieve around \(l=0.8\) in the trace of Facebook cluster A in 2024.

\noindent \textbf{Result measurement.} We measure the performance of our algorithm in terms of running time per operation and number of rewirings per operation. Running time per operation is the total running time divided by the number of atomic logical topology changes, in nanoseconds. Number of rewirings per operation is the total number of rewirings divided by the number of atomic logical topology changes.

%-----------------------------------
\subsection{Evaluation Results}
%-----------------------------------

\begin{table}[!t]
\caption{Real-time ToE evaluation results under different \(c\)'s.}\label{tab4}
\begin{center}
\begin{threeparttable}

\begin{tabular}{c|ccc|ccc}
\toprule

\multirow{2}{*}{$(n,c)$}&\multirow{2}{*}{$l$}&Time\tnote{*}&\multirow{2}{*}{NR\tnote{**}}&\multirow{2}{*}{$l$}&Time&\multirow{2}{*}{NR}\\
&&(ns)&&&(ns)&\\

\midrule[0.8pt]

\multirow{2.5}{*}{$(256,4)$}&0.2&18.72&0.023&0.4&119.5&0.179\\
\cmidrule{2-7}
&0.6&494.9&0.744&0.8&1625&1.406\\

\midrule

\multirow{2.5}{*}{$(256,8)$}&0.2&17.84&0.024&0.4&119.0&0.167\\
\cmidrule{2-7}
&0.6&419.2&0.694&0.8&1237&1.290\\

\midrule
\multirow{2.5}{*}{$(256,16)$}&0.2&16.21&0.024&0.4&102.3&0.161\\
\cmidrule{2-7}
&0.6&375.3&0.669&0.8&960.4&1.209\\

\bottomrule
\end{tabular}

\begin{tablenotes}
    \item[*] \emph{Time} refers to running time per operation.\\
    \item[**] \emph{NR} refers to number of rewirings per operation.
\end{tablenotes}

\end{threeparttable}
\end{center}
\end{table}

\begin{table}[!t]
\caption{Real-time ToE evaluation results under different \(n\)'s.}\label{tab5}
\begin{center}
\begin{threeparttable}

\begin{tabular}{c|ccc|ccc}
\toprule

\multirow{2}{*}{$(n,c)$}&\multirow{2}{*}{$l$}&Time\tnote{*}&\multirow{2}{*}{NR\tnote{**}}&\multirow{2}{*}{$l$}&Time&\multirow{2}{*}{NR}\\
&&(ns)&&&(ns)&\\

\midrule[0.8pt]

\multirow{2.5}{*}{$(128,8)$}&0.2&16.80&0.022&0.4&106.4&0.178\\
\cmidrule{2-7}
&0.6&380.1&0.717&0.8&1077&1.313\\

\midrule
\multirow{2.5}{*}{$(256,8)$}&0.2&17.84&0.024&0.4&119.0&0.167\\
\cmidrule{2-7}
&0.6&419.2&0.694&0.8&1237&1.290\\

\midrule
\multirow{2.5}{*}{$(384,8)$}&0.2&17.97&0.024&0.4&130.4&0.162\\
\cmidrule{2-7}
&0.6&455.9&0.685&0.8&1373&1.275\\

\bottomrule
\end{tabular}

\begin{tablenotes}
    \item[*] \emph{Time} refers to running time per operation.\\
    \item[**] \emph{NR} refers to number of rewirings per operation.
\end{tablenotes}

\end{threeparttable}
\end{center}
\end{table}

The results are shown in Table~\ref{tab4} and Table~\ref{tab5}. When $l=0.2$, our algorithm only requires less than 20 nanoseconds to process an operation of atomic logical topology change. However, this number increases considerably as $l$ increases. When $l=0.8$, it takes an average of about 1000 nanoseconds. The trend of the number of rewirings is similar: when $l=0.2$, it only requires around 0.023 rearrangements per operation, and when $l=0.8$, this number increases to more than 1. It should be noted that the number of rewirings per operation when $l=0.2$ is much lower than 1, showing that our algorithm does a good job maintaining redundant connections. For other parameters, from Table~\ref{tab4} we can see that as $c$ increases, both the running time per operation and the number of rewirings per operation decrease. From Table~\ref{tab5} we can see that as $n$ increases, the running time per operation increases, while the number of rewirings per operation decreases, but both differences are not significant.
%-------------------------------------------------------------------------------

%-------------------------------------------------------------------------------
\section{Related Work}\label{relwork}
%-------------------------------------------------------------------------------

In fact, the replacement chain is not a completely new concept. Paull proposed a similar algorithm with linear time complexity, but it was limited to replacing between two intermediate switches in networks with link capacities equal to 1 \cite{paull1962reswitching}. In this paper, we generalize the algorithm to networks with integer link capacities greater than 1 (in Section~\ref{resttwo}) while still maintaining linear time complexity. We also generalize the algorithm to replacing between multiple intermediate switches/OCes to reduce the number of rewirings, and experimentally show that the efficiency of this algorithm in practice is acceptable, although no strict upper bound on the time complexity is obtained. There have been studies on the upper bound of the number of rewirings of Paull's algorithm, such as the work of Satoru Ohta \cite{ohta2020number,ohta2021number} and the work of Frank K. Hwang and Wen-Dar Lin \cite{hwang1998number}, which may be helpful in proving a strict upper bound on the time complexity for this algorithm.

For the minimum rewiring problem, Google \cite{zhao2019minimal, Wang_Zhang_Li_2022} and Shuyuan Zhang et al. \cite{zhang2023reducing} have published some algorithms. These algorithms have been briefly described in Section~\ref{evalalg}. However, the biggest problem with these algorithms is that they all start by modeling the problem as an off-the-shelf ILP or MCF model. This idea limits the algorithm to only solving the complete problem, and it is impossible to adjust only the local part of the matching scheme, making it unsuitable for small logical topology changes. Moreover, this idea does not fully utilize the characteristics of the problem itself. In contrast, our algorithm starts from the actual form of the problem, fully utilizes the characteristics of the problem itself, and divides the solution of the problem into the solution of multiple local atomic problems, making it suitable for small logical topology changes, meanwhile supporting real-time ToE which none of the works solved before.

%-------------------------------------------------------------------------------
\section{Conclusion and Future Work}\label{conclusion}
%-------------------------------------------------------------------------------

In this paper, we propose the novel \emph{bidirectional model} for the ToE problem for the first time, together with the \emph{FastReChain} algorithm based on the model. Experiments show that the algorithm not only bypassed the network adaptation loss between the bidirectional and traditional model incurred in previous works, but also completely outperformed prior algorithms across all metrics in medium- to large-scale clusters, even in the traditional model. Additionally, we also achieved real-time ToE for the first time for potentially future optical data centers, although there is no practical application of this type of task to date.

It is worth noting that the superior performance of \emph{FastReChain} is not unfounded, but based on our modeling design, solid theoretical analysis, and sophisticated algorithmic ideas. First, we provide rigorous proof that the algorithm is always capable of finding feasible solutions in proportional models. Second, we conducted careful analysis, speculation, and observation about the efficiency of the algorithm, and consequently propose the bitset optimization in Section~\ref{bsopt} to massively speed up the solving efficiency, and the gradual limitation of search breadth in Section~\ref{adap} to add a layer of protection for extreme cases.

Nonetheless, we still have a lot of ideas in further improving the algorithm, and we leave them for future works:

\textcircled{1} Since each replacement chain only affects a small number of connections, \textbf{this algorithm has great potential for parallel acceleration}. We believe this is especially important in real-time ToE where time efficiency is extremely critical.

\textcircled{2} The preformance of the algorithm greatly relies on the bitset optimization, but bitsets typically only supports operations on 64 bits at a time. In larger clusters, e.g. thousands of switches, we might need the support for more bits at a time using technologies such as SIMD.

\textcircled{3} There might be room of improvements in terms of the algorithm itself, e.g. achieving better robustness or rewiring ratio by involving the 3\textsuperscript{rd} case mentioned in \ref{formrc}; pruning useless searches; or a better refining strategy than that in Section~\ref{refinedver}.

% In this paper, we present a centralized scheduling algorithm that achieves dynamic scheduling for the first time, while outperforming existing algorithms in terms of static minimal rewiring on all metrics. We prove that the algorithm can always achieve the maximum throughput in proportional Clos networks, and demonstrate the efficiency of the algorithm through experiments. With the support of our algorithm, the community can start designing highly responsive optical switching network architectures based on dynamic centralized scheduling, which may bring a promising future for optical networking. Even for static minimal rewiring, our algorithm provides fewer rewirings and less running time than existing algorithms, thus supporting more frequent network adjustments and less impact on ongoing communications.

% %-------------------------------------------------------------------------------
% \section*{Acknowledgments}
% %-------------------------------------------------------------------------------

% The USENIX latex style is old and very tired, which is why
% there's no \textbackslash{}acks command for you to use when
% acknowledging. Sorry.

%-------------------------------------------------------------------------------
\section*{Availability}
%-------------------------------------------------------------------------------

The implementation and evaluation code of all the algorithms mentioned above is available in the C++ language as an open source project in \cite{fastrechaingithub}. We aim to develop a high-quality open source library of ToE algorithms for the development of future network architectures.

\bibliographystyle{IEEEtran}
\bibliography{main}

\clearpage

\appendices

%-------------------------------------------------------------------------------
\section{An Example of Bipartite Divide-and-Conquer Loss}\label{biploss}
%-------------------------------------------------------------------------------
The following is an example of bipartite divide-and-conquer loss, in a traditional model with $n=4$, $m=4$, and $C_{*,*}=1$. The matrices in the first rows of each scheme are $X_{0},X_{1},X_{2},X_{3}$ (or $Y_{0},\ldots$), and the matrices in the second rows are the element-wise sums of the above 2 matrices. NR represents \textit{number of rewirings}.

We can see that in the first bipartition step (the second row), the matching schemes are optimal in terms of NR (with \(\mathrm{NR}=8\)). However, the optimality is not further guaranteed in the final step (the first row).
\begin{center}

Target logical topology:
$\begin{bmatrix}
    \colorbox{danhong}{1}\colorbox{danhong}{1}\colorbox{danhong}{1}\colorbox{danhong}{1}\\
    \colorbox{danhong}{1}\colorbox{danhong}{1}\colorbox{danhong}{1}\colorbox{danhong}{1}\\
    \colorbox{danhong}{1}\colorbox{danhong}{1}\colorbox{danhong}{1}\colorbox{danhong}{1}\\
    \colorbox{danhong}{1}\colorbox{danhong}{1}\colorbox{danhong}{1}\colorbox{danhong}{1}
\end{bmatrix}$

\ \\

Initial matching scheme:
$\begin{bmatrix}
    \colorbox{danlan}{0}\colorbox{danhong}{1}\colorbox{danlan}{0}\colorbox{danlan}{0}\\
    \colorbox{danlan}{0}\colorbox{danlan}{0}\colorbox{danhong}{1}\colorbox{danlan}{0}\\
    \colorbox{danhong}{1}\colorbox{danlan}{0}\colorbox{danlan}{0}\colorbox{danlan}{0}\\
    \colorbox{danlan}{0}\colorbox{danlan}{0}\colorbox{danlan}{0}\colorbox{danhong}{1}
\end{bmatrix}
\begin{bmatrix}
    \colorbox{danlan}{0}\colorbox{danlan}{0}\colorbox{danlan}{0}\colorbox{danhong}{1}\\
    \colorbox{danlan}{0}\colorbox{danhong}{1}\colorbox{danlan}{0}\colorbox{danlan}{0}\\
    \colorbox{danlan}{0}\colorbox{danlan}{0}\colorbox{danhong}{1}\colorbox{danlan}{0}\\
    \colorbox{danhong}{1}\colorbox{danlan}{0}\colorbox{danlan}{0}\colorbox{danlan}{0}
\end{bmatrix}
\begin{bmatrix}
    \colorbox{danlan}{0}\colorbox{danlan}{0}\colorbox{danhong}{1}\colorbox{danlan}{0}\\
    \colorbox{danlan}{0}\colorbox{danlan}{0}\colorbox{danlan}{0}\colorbox{danhong}{1}\\
    \colorbox{danhong}{1}\colorbox{danlan}{0}\colorbox{danlan}{0}\colorbox{danlan}{0}\\
    \colorbox{danlan}{0}\colorbox{danhong}{1}\colorbox{danlan}{0}\colorbox{danlan}{0}
\end{bmatrix}
\begin{bmatrix}
    \colorbox{danlan}{0}\colorbox{danhong}{1}\colorbox{danlan}{0}\colorbox{danlan}{0}\\
    \colorbox{danhong}{1}\colorbox{danlan}{0}\colorbox{danlan}{0}\colorbox{danlan}{0}\\
    \colorbox{danlan}{0}\colorbox{danlan}{0}\colorbox{danhong}{1}\colorbox{danlan}{0}\\
    \colorbox{danlan}{0}\colorbox{danlan}{0}\colorbox{danlan}{0}\colorbox{danhong}{1}
\end{bmatrix}$

$\searrow+\swarrow$\ \ \ \ \ \ \ \ \ \ \ \ \ \ \ \ \ \ \ \ \ \ \ \ $\searrow+\swarrow$

$\begin{bmatrix}
    \colorbox{danlan}{0}\colorbox{danhong}{1}\colorbox{danlan}{0}\colorbox{danhong}{1}\\
    \colorbox{danlan}{0}\colorbox{danhong}{1}\colorbox{danhong}{1}\colorbox{danlan}{0}\\
    \colorbox{danhong}{1}\colorbox{danlan}{0}\colorbox{danhong}{1}\colorbox{danlan}{0}\\
    \colorbox{danhong}{1}\colorbox{danlan}{0}\colorbox{danlan}{0}\colorbox{danhong}{1}
\end{bmatrix}
\ \ \ \ \ \ \ \ \ \ \ \ \ \ \ \ 
\begin{bmatrix}
    \colorbox{danlan}{0}\colorbox{danhong}{1}\colorbox{danhong}{1}\colorbox{danlan}{0}\\
    \colorbox{danhong}{1}\colorbox{danlan}{0}\colorbox{danlan}{0}\colorbox{danhong}{1}\\
    \colorbox{danhong}{1}\colorbox{danlan}{0}\colorbox{danhong}{1}\colorbox{danlan}{0}\\
    \colorbox{danlan}{0}\colorbox{danhong}{1}\colorbox{danlan}{0}\colorbox{danhong}{1}
\end{bmatrix}$

\ \\

Optimal matching scheme ($\mathrm{NR}=8$):
$\begin{bmatrix}
    {\color{red}\colorbox{danhong}{\textbf{1}}}{\color{red}\colorbox{danlan}{\textbf{0}}}\colorbox{danlan}{0}\colorbox{danlan}{0}\\
    \colorbox{danlan}{0}\colorbox{danlan}{0}\colorbox{danhong}{1}\colorbox{danlan}{0}\\
    {\color{red}\colorbox{danlan}{\textbf{0}}}{\color{red}\colorbox{danhong}{\textbf{1}}}\colorbox{danlan}{0}\colorbox{danlan}{0}\\
    \colorbox{danlan}{0}\colorbox{danlan}{0}\colorbox{danlan}{0}\colorbox{danhong}{1}
\end{bmatrix}
\begin{bmatrix}
    \colorbox{danlan}{0}\colorbox{danlan}{0}\colorbox{danlan}{0}\colorbox{danhong}{1}\\
    \colorbox{danlan}{0}\colorbox{danhong}{1}\colorbox{danlan}{0}\colorbox{danlan}{0}\\
    \colorbox{danlan}{0}\colorbox{danlan}{0}\colorbox{danhong}{1}\colorbox{danlan}{0}\\
    \colorbox{danhong}{1}\colorbox{danlan}{0}\colorbox{danlan}{0}\colorbox{danlan}{0}
\end{bmatrix}
\begin{bmatrix}
    \colorbox{danlan}{0}\colorbox{danlan}{0}\colorbox{danhong}{1}\colorbox{danlan}{0}\\
    \colorbox{danlan}{0}\colorbox{danlan}{0}\colorbox{danlan}{0}\colorbox{danhong}{1}\\
    \colorbox{danhong}{1}\colorbox{danlan}{0}\colorbox{danlan}{0}\colorbox{danlan}{0}\\
    \colorbox{danlan}{0}\colorbox{danhong}{1}\colorbox{danlan}{0}\colorbox{danlan}{0}
\end{bmatrix}
\begin{bmatrix}
    \colorbox{danlan}{0}\colorbox{danhong}{1}\colorbox{danlan}{0}\colorbox{danlan}{0}\\
    \colorbox{danhong}{1}\colorbox{danlan}{0}\colorbox{danlan}{0}\colorbox{danlan}{0}\\
    \colorbox{danlan}{0}\colorbox{danlan}{0}{\color{red}\colorbox{danlan}{\textbf{0}}}{\color{red}\colorbox{danhong}{\textbf{1}}}\\
    \colorbox{danlan}{0}\colorbox{danlan}{0}{\color{red}\colorbox{danhong}{\textbf{1}}}{\color{red}\colorbox{danlan}{\textbf{0}}}
\end{bmatrix}$

$\searrow+\swarrow$\ \ \ \ \ \ \ \ \ \ \ \ \ \ \ \ \ \ \ \ \ \ \ \ $\searrow+\swarrow$

$\begin{bmatrix}
    {\color{red}\colorbox{danhong}{\textbf{1}}}{\color{red}\colorbox{danlan}{\textbf{0}}}\colorbox{danlan}{0}\colorbox{danhong}{1}\\
    \colorbox{danlan}{0}\colorbox{danhong}{1}\colorbox{danhong}{1}\colorbox{danlan}{0}\\
    {\color{red}\colorbox{danlan}{\textbf{0}}}{\color{red}\colorbox{danhong}{\textbf{1}}}\colorbox{danhong}{1}\colorbox{danlan}{0}\\
    \colorbox{danhong}{1}\colorbox{danlan}{0}\colorbox{danlan}{0}\colorbox{danhong}{1}
\end{bmatrix}
\ \ \ \ \ \ \ \ \ \ \ \ \ \ \ \ 
\begin{bmatrix}
    \colorbox{danlan}{0}\colorbox{danhong}{1}\colorbox{danhong}{1}\colorbox{danlan}{0}\\
    \colorbox{danhong}{1}\colorbox{danlan}{0}\colorbox{danlan}{0}\colorbox{danhong}{1}\\
    \colorbox{danhong}{1}\colorbox{danlan}{0}{\color{red}\colorbox{danlan}{\textbf{0}}}{\color{red}\colorbox{danhong}{\textbf{1}}}\\
    \colorbox{danlan}{0}\colorbox{danhong}{1}{\color{red}\colorbox{danhong}{\textbf{1}}}{\color{red}\colorbox{danlan}{\textbf{0}}}
\end{bmatrix}$

\ \\

MCF matching scheme ($\mathrm{NR}=16$):
$\begin{bmatrix}
    \colorbox{danlan}{0}\colorbox{danhong}{1}\colorbox{danlan}{0}\colorbox{danlan}{0}\\
    \colorbox{danlan}{0}\colorbox{danlan}{0}\colorbox{danhong}{1}\colorbox{danlan}{0}\\
    {\color{red}\colorbox{danlan}{\textbf{0}}}\colorbox{danlan}{0}\colorbox{danlan}{0}{\color{red}\colorbox{danhong}{\textbf{1}}}\\
    {\color{red}\colorbox{danhong}{\textbf{1}}}\colorbox{danlan}{0}\colorbox{danlan}{0}{\color{red}\colorbox{danlan}{\textbf{0}}}
\end{bmatrix}
\begin{bmatrix}
    \colorbox{danlan}{0}\colorbox{danlan}{0}\colorbox{danlan}{0}\colorbox{danhong}{1}\\
    \colorbox{danlan}{0}\colorbox{danhong}{1}\colorbox{danlan}{0}\colorbox{danlan}{0}\\
    {\color{red}\colorbox{danhong}{\textbf{1}}}\colorbox{danlan}{0}{\color{red}\colorbox{danlan}{\textbf{0}}}\colorbox{danlan}{0}\\
    {\color{red}\colorbox{danlan}{\textbf{0}}}\colorbox{danlan}{0}{\color{red}\colorbox{danhong}{\textbf{1}}}\colorbox{danlan}{0}
\end{bmatrix}
\begin{bmatrix}
    {\color{red}\colorbox{danhong}{\textbf{1}}}\colorbox{danlan}{0}{\color{red}\colorbox{danlan}{\textbf{0}}}\colorbox{danlan}{0}\\
    \colorbox{danlan}{0}\colorbox{danlan}{0}\colorbox{danlan}{0}\colorbox{danhong}{1}\\
    {\color{red}\colorbox{danlan}{\textbf{0}}}\colorbox{danlan}{0}{\color{red}\colorbox{danhong}{\textbf{1}}}\colorbox{danlan}{0}\\
    \colorbox{danlan}{0}\colorbox{danhong}{1}\colorbox{danlan}{0}\colorbox{danlan}{0}
\end{bmatrix}
\begin{bmatrix}
    \colorbox{danlan}{0}{\color{red}\colorbox{danlan}{\textbf{0}}}{\color{red}\colorbox{danhong}{\textbf{1}}}\colorbox{danlan}{0}\\
    \colorbox{danhong}{1}\colorbox{danlan}{0}\colorbox{danlan}{0}\colorbox{danlan}{0}\\
    \colorbox{danlan}{0}{\color{red}\colorbox{danhong}{\textbf{1}}}{\color{red}\colorbox{danlan}{\textbf{0}}}\colorbox{danlan}{0}\\
    \colorbox{danlan}{0}\colorbox{danlan}{0}\colorbox{danlan}{0}\colorbox{danhong}{1}
\end{bmatrix}$

$\searrow+\swarrow$\ \ \ \ \ \ \ \ \ \ \ \ \ \ \ \ \ \ \ \ \ \ \ \ $\searrow+\swarrow$

$\begin{bmatrix}
    \colorbox{danlan}{0}\colorbox{danhong}{1}\colorbox{danlan}{0}\colorbox{danhong}{1}\\
    \colorbox{danlan}{0}\colorbox{danhong}{1}\colorbox{danhong}{1}\colorbox{danlan}{0}\\
    \colorbox{danhong}{1}\colorbox{danlan}{0}{\color{red}\colorbox{danlan}{\textbf{0}}}{\color{red}\colorbox{danhong}{\textbf{1}}}\\
    \colorbox{danhong}{1}\colorbox{danlan}{0}{\color{red}\colorbox{danhong}{\textbf{1}}}{\color{red}\colorbox{danlan}{\textbf{0}}}
\end{bmatrix}
\ \ \ \ \ \ \ \ \ \ \ \ \ \ \ \ 
\begin{bmatrix}
    {\color{red}\colorbox{danhong}{\textbf{1}}}{\color{red}\colorbox{danlan}{\textbf{0}}}\colorbox{danhong}{1}\colorbox{danlan}{0}\\
    \colorbox{danhong}{1}\colorbox{danlan}{0}\colorbox{danlan}{0}\colorbox{danhong}{1}\\
    {\color{red}\colorbox{danlan}{\textbf{0}}}{\color{red}\colorbox{danhong}{\textbf{1}}}\colorbox{danhong}{1}\colorbox{danlan}{0}\\
    \colorbox{danlan}{0}\colorbox{danhong}{1}\colorbox{danlan}{0}\colorbox{danhong}{1}
\end{bmatrix}$

\ \\

(Differences from the initial scheme are marked in {\color{red}\textbf{0}} or {\color{red}\textbf{1}})
\end{center}
% \end{figure}

\newpage
%-------------------------------------------------------------------------------
\section{Proof of the Logical Topology Conversion Algorithm}\label{nsaproof}
%-------------------------------------------------------------------------------

Since \(G_{j,k}\) is the remainder of \(C_{j,k}\) divided by \(2\), it can be represented as
\begin{equation}
G_{j,k}=D_{j,k}-2\left\lfloor\frac12D_{j,k}\right\rfloor
\end{equation}
Since all edges connected to each vertex in the graph are oriented in either direction, and the difference between the number of edges in the two directions does not exceed 1, this can be represented as
\begin{subequations}
\begin{align}
\forall j,&\left\lfloor\frac 12\sum_kG_{j,k}\right\rfloor\le\sum_kG'_{j,k}\le\left\lceil\frac 12\sum_kG_{j,k}\right\rceil\\
\forall k,&\left\lfloor\frac 12\sum_jG_{j,k}\right\rfloor\le\sum_jG'_{j,k}\le\left\lceil\frac 12\sum_jG_{j,k}\right\rceil
\end{align}
\end{subequations}
Then we have
\begin{equation}
\begin{aligned}
\sum_kD'_{j,k}&=\sum_k\left\lfloor\frac12D_{j,k}\right\rfloor+\sum_kG'_{j,k}\\
&\ge\sum_k\left\lfloor\frac12D_{j,k}\right\rfloor+\left\lfloor\frac 12\sum_kG_{j,k}\right\rfloor\\
&=\sum_k\left\lfloor\frac12D_{j,k}\right\rfloor+\left\lfloor\frac 12\sum_kD_{j,k}-\sum_k\left\lfloor\frac12D_{j,k}\right\rfloor\right\rfloor\\
&=\sum_k\left\lfloor\frac12D_{j,k}\right\rfloor+\left\lfloor\frac 12\sum_kD_{j,k}\right\rfloor-\sum_k\left\lfloor\frac12D_{j,k}\right\rfloor\\
&=\left\lfloor\frac 12\sum_kD_{j,k}\right\rfloor
\end{aligned}
\end{equation}
The same goes for other constraints.

\newpage
\onecolumn
%-------------------------------------------------------------------------------
\section{Pseudocode for Replacement Chain IDS}\label{rcids}
%-------------------------------------------------------------------------------

\begin{algorithm}
    \SetKwProg{Fn}{function}{:}{end}
    \caption{Replacement chain IDS.}
    
    \KwData{$m$, $n$, $C$, $D$, current scheme $X$, low-level switch id $j_0$ and $k_0$ between which the connection is to be scheduled}
    \KwResult{whether the connection is successfully scheduled. If true, $X$ becomes the new scheme}
    \BlankLine
    \begin{multicols}{2}
    
        \SetKwFunction{AddConnection}{AddConnection}
        \Fn{\AddConnection{$i$, $j$, $k$}} {
            $X_{i,j,k}\gets X_{i,j,k}+1$, $X_{i,k,j}\gets X_{i,k,j}+1$\;
        }
        \BlankLine
    
        \SetKwFunction{RemoveConnection}{RemoveConnection}
        \Fn{\RemoveConnection{$i$, $j$, $k$}} {
            $X_{i,j,k}\gets X_{i,j,k}-1$, $X_{i,k,j}\gets X_{i,k,j}-1$\;
        }
        \BlankLine
    
        \SetKwFunction{FindRedundantConnection}{FindRedundantConnection}
        \Fn{\FindRedundantConnection{$i$, $j$}} {
            \textbf{if} $\sum_{k=0}^{m-1}X_{i,j,k}<C_{i,j}$ \textbf{then}\\
            \quad\ \,\KwRet{$m$}\;
            \ForEach{$k\in\{0,1,\ldots,m-1\}$ in random order} {
                \textbf{if} $X_{i,j,k}>0\land E(X)_{j,k}>D_{j,k}$ \textbf{then}\\
                \quad\ \,\KwRet{$k$}\;
            }
            \KwRet{$-1$}\;
        }
        \BlankLine

        \SetKwFunction{DepthLimitedSearch}{DepthLimitedSearch}
        \SetKwFunction{DLSSubroutine}{DLSSubroutine}
        
        \Fn{\DLSSubroutine{$i$, $j$, $k$, $jConn$, $depth$}} {
            \textbf{if} $jConn=-1$ \textbf{then}\\
            \quad\ \,\KwRet{false}\;
            \textbf{if} $jConn<m$ \textbf{then}\\
            \quad\ \,\RemoveConnection{$i$, $j$, $jConn$}\;
            \ForEach{$l\in\{0,1,\ldots,m-1\}-\{j\}$ in random order \qquad\qquad} {
                \If{$X_{i,k,l}>0$} {
                    \RemoveConnection{$i$, $k$, $l$}\;
                    \textbf{if} \DepthLimitedSearch{$k$, $l$, $depth-1$} \textbf{then}\\
                    \quad\ \,\KwRet{true}\;
                    \AddConnection{$i$, $k$, $l$}\;
                }
            }
            \textbf{if} $jConn<m$ \textbf{then}\\
            \quad\ \,\AddConnection{$i$, $j$, $jConn$}\;
            \KwRet{false}\;
        }
        
        \Fn{\DepthLimitedSearch{$j$, $k$, $depth$}} {
            \eIf{$depth=0$} {
                \ForEach{$i\in\{0,1,\ldots,n-1\}$ in random order \qquad\qquad} {
                    $jConn\gets$ \FindRedundantConnection{$i$, $j$}\;
                    $kConn\gets$ \FindRedundantConnection{$i$, $k$}\;
                    \If{$jConn\ne-1\land kConn\ne-1$} {
                        \textbf{if} $jConn<m$ \textbf{then}\\
                        \quad\ \,\RemoveConnection{$i$, $j$, $jConn$}\;
                        \textbf{if} $kConn<m$ \textbf{then}\\
                        \quad\ \,\RemoveConnection{$i$, $k$, $kConn$}\;
                        \AddConnection{$i$, $j$, $k$}\;
                        \KwRet{true}\;
                    }
                }
            } {
                \ForEach{$i\in\{0,1,\ldots,n-1\}$ in random order \qquad\qquad} {
                    \AddConnection{$i$, $j$, $k$}\;
                    $jConn\gets$ \FindRedundantConnection{$i$, $j$}\;
                    $kConn\gets$ \FindRedundantConnection{$i$, $k$}\;
                    \textbf{if} \DLSSubroutine{$i$, $j$, $k$, $jConn$, $depth$} $\lor$ \DLSSubroutine{$i$, $k$, $j$, $kConn$, $depth$} \textbf{then}\\
                    \quad\ \,\KwRet{true}\;
                    \RemoveConnection{$i$, $j$, $k$}\;
                }
            }
            \KwRet{false}\;
        }
        \BlankLine

        \For{$depth\gets0$ \KwTo depth limit (e.g., $\sum_{j=0}^{m-1}W_{L_j}-1$)} {
            \textbf{if} \DepthLimitedSearch{$j_0$, $k_0$, $depth$} \textbf{then}\\
            \quad\ \,\KwRet{true}\;
        }
        \KwRet{false}\;
    \end{multicols}
    \BlankLine
\end{algorithm}

 

% \end{IEEEbiography}



% \vfill

\end{document}